\documentclass[lettersize,journal]{IEEEtran}
\usepackage{amsmath,amsfonts}
\usepackage{algorithmic}
\usepackage{algorithm}
\usepackage{array}
\usepackage[caption=false,font=normalsize,labelfont=sf,textfont=sf]{subfig}
\usepackage{textcomp}
\usepackage{stfloats}
\usepackage{url}
\usepackage{verbatim}
\usepackage{graphicx}
\usepackage{cite}
\usepackage{amsfonts,color,morefloats,pslatex,a4wide}
\usepackage{amssymb,amsthm,amsmath,latexsym,pslatex,cite}
\usepackage{lineno,mathtools}
\usepackage{pdflscape}
\usepackage[colorlinks, linkcolor=red, citecolor=green]{hyperref}

\newtheorem{theorem}{Theorem}[section]

\newtheorem{problem}{Problem}
\newtheorem{lemma}[theorem]{Lemma}
\newtheorem{corollary}[theorem]{Corollary}

\theoremstyle{definition} % 使用定义样式，使字体为正体
\newtheorem{remark}{Remark}
\newtheorem{example}{Example}
\newtheorem{definition}{Definition}
\newcounter{MYtempeqncnt}

\newcommand{\F}{\mathbb{F}}

\newcommand{\w}{\omega}

\newcommand{\C}{{\mathcal{C}}}
\newcommand{\D}{{\mathcal{D}}}
\newcommand{\SSS}{{\mathcal{S}}}

\newcommand{\GRS}{{\mathrm{GRS}}}
\newcommand{\EGRS}{{\mathrm{EGRS}}}

\newcommand{\bc}{{{\bf c}}}
\newcommand{\bx}{{{\bf x}}}
\newcommand{\by}{{{\bf y}}}

\newcommand{\aaa}{{{\mathbf{a}}}}

\newcommand{\wt}{{{\rm{wt}}}}

\newcommand{\rank}{{\rm rank}}

\makeatletter

\newcommand{\Rmnum}[1]{\expandafter\@slowromancap\romannumeral #1@}
\makeatother

\begin{document}

\title{A family of linear codes that are either non-GRS MDS codes or NMDS codes}
\author{Yang Li, Zhonghua Sun, Shixin Zhu 
\thanks{Yang Li, Zhonghua Sun, and Shixin Zhu are with School of Mathematics, Hefei University of Technology, Hefei, 230601, Anhui, China. 
(Emails: yanglimath@163.com, sunzhonghuas@163.com, and zhushixinmath@hfut.edu.cn.)}
\thanks{
This research is supported by 
the National Natural Science Foundation of China under Grant No. U21A20428, 
the Anhui Provincial Natural Science Foundation under Grant No. 2408085MA014,  
the Fundamental Research Funds for the Central Universities under Grant No. JZ2024HGTB0203, and 
the Natural Science Research Key Project of Anhui Educational Committee under Grant No. 2024AH050303.} 
\thanks{Manuscript received -, -; revised -, -.}
}


\maketitle

\begin{abstract}
  Both maximum distance separable (MDS) codes that are not equivalent to 
  generalized Reed-Solomon (GRS) codes (non-GRS MDS codes) and near MDS (NMDS) codes have nice applications 
  in communication and storage systems. 
  In this paper, we introduce and study a new family of linear codes involving their parameters, 
  weight distributions, and self-orthogonal properties. 
  We prove that such codes are either non-GRS MDS codes or NMDS codes, and hence, 
  they can produce as many of the desired codes as possible.   
  We also completely determine their weight distributions with the help of the solutions to some subset sum problems. 
  A sufficient and necessary condition for such codes to be self-orthogonal is characterized. 
  Based on this condition, we further deduce that there are no self-dual codes 
  in this class of linear codes and explicitly construct two new classes of almost self-dual codes. 
\end{abstract}


\section{Introduction}\label{sec.introduction}

Let $q=p^m$ and $\F_q$ be the finite field of size $q$, where $p$ is a prime and $m$ is a positive integer. 
Let $\F_q^*=\F_q\setminus \{0\}$ be the multiplicative group of $\F_q$.
A $q$-ary {\em linear code} $\C$ of length $n$ and dimension $k$, denoted by $[n,k]_q$,
is a $k$-dimensional linear subspace of $\F_q^n$. If $\C$ has a minimum distance $d$,
we further denote $\C$ by $[n,k,d]_q$. 
For any two vectors $\mathbf{x}=(x_1,\ldots,x_n)$ and $\mathbf{y}=(y_1,\ldots,y_n)$ in $\F_q^n$, 
define their {\em Euclidean inner product} as 
$
  \langle \mathbf{x}, \mathbf{y} \rangle=\sum_{i=1}^{n}x_iy_i. 
$
Then the {\em Euclidean dual code} of an $[n,k]_q$ linear code $\C$ 
is an $[n,n-k]_q$ linear code given by 
$
  \C^{\perp}=\left\{\mathbf{y}\in \F_q^n: \langle \mathbf{x},\mathbf{y} \rangle=0, \forall~ \mathbf{x}\in \C \right\}.
$
We call $\C$ a {\em self-orthogonal} code if $\C\subseteq \C^{\perp}$; 
a {\em self-dual} code if $\C=\C^{\perp}$; and an {\em almost self-dual} code if $\C$ is an $[n,\frac{n-1}{2}]_q$ self-orthogonal code for odd $n$. 
% Self-orthogonal codes are of particular interest. 

For any $[n,k,d]_q$ linear code $\C$, the parameters $n$, $k$ and $d$ are subject to the well-known {\em Singleton bound}, 
which says: $d\leq n-k+1$ \cite{HP2003}. 
If $d=n-k+1$, we call $\C$ a {\em maximum distance separable (MDS)} code and 
if $d=n-k$, we call $\C$ an {\em almost MDS (AMDS)} code. 
Moreover, $\C$ is called a {\em near MDS (NMDS)} code if both $\C$ and $\C^{\perp}$ are AMDS codes. 
MDS codes and NMDS codes have attracted lots of attention because of their important applications in 
distributed storage systems \cite{CHL2011}, combinatorial designs \cite{DT2020}, finite geometry \cite{DL1994}, 
random error channels \cite{ST2013}, informed source and index coding problems \cite{TR2018}, and secret sharing scheme \cite{SV2018,ZWXLQY2009}.

Although the well-known generalized Reed-Solomon (GRS) codes (a special class of MDS codes \cite{HP2003}) can provide almost 
all desired parameters of MDS codes, {it is still very interesting to construct 
MDS codes that are not equivalent to GRS codes, referred as non-GRS  MDS codes},  
since the non-GRS properties make them resistant to Sidelnikov-Shestakov attacks and Wieschebrink attacks in cryptography systems, whereas GRS codes are not \cite{LR2020,BBPR2018}. 
Much progress has been made on this topic and we can summarize some of them as follows: 
\begin{itemize}
    \item By adding two columns in generator matrices of GRS codes, Roth and Lempel constructed the so-called Roth-Lempel MDS codes and proved that 
    they are non-GRS \cite{RL1989}. Note that Han and Fan have also recently studied the NMDS properties of Roth-Lempel codes \cite{HF2023}. 

    \item By adding twists in GRS codes, Beelen $et~al.$ introduced twisted GRS (TGRS) codes \cite{BPR2022}. 
    Following this work, the authors further studied many different twists and showed that most of the corresponding TGRS codes can generate non-GRS  MDS codes 
    or NMDS codes (see \cite{GLLS2023,BPR2022,ZZT2022,HYNYL2021,SYLH2022,SYS2023,W2021,ZL2024,WHL2021,ll2021,C2023,GZ223} and the references therein). 
    Note also that many self-orthogonal and (almost) self-dual codes have been constructed in these papers. 
    %It should be emphasized that Zhu and Liao \cite{ZL2024} have investigated the class of (+)-extended TGRS codes. 

    \item By adding an overall parity-check in GRS codes, He and Liao obtained the lengthened GRS (LGRS) codes \cite{HL2023}. 
    They obtained some sufficient and necessary conditions for an LGRS code to be MDS or AMDS and presented some sufficient 
    conditions for which an LGRS code is non-GRS.   

    \item Very recently, by deleting the penultimate row of generator matrices of GRS codes, 
    Han and Zhang constructed a class of $q$-ary linear codes and proved that 
    they are either MDS or NMDS without analysing their non-GRS properties \cite{HZ2023}.  
    Furthermore, they obtained many self-dual NMDS codes. 
    For more constructions of MDS and NMDS codes, readers can also refer to the papers \cite{JK2019,HLW2022,SD2023,CHao2023}. 
\end{itemize}

{
On one hand, it is easily seen from the above discussions that most of the known non-GRS  MDS codes 
and NMDS codes are from the class of TGRS codes. 
On the other hand, in consideration of the important applications of MDS codes and NMDS codes, 
it is very interesting to construct a family of linear codes such that it produces non-GRS  MDS codes and NMDS codes simultaneously. 
Furthermore, it is of particular interest in constructing a class of linear codes such that they are 
either non-GRS  MDS codes or NMDS codes, as this will produce as many of the desired codes as possible.
Some representative works related to this topic can be found in \cite{HZ2023,ZL2024,HYNYL2021,HF2023,RL1989}.
Thus, {\bf the first motivation} of this paper is to address the following problem:
\begin{problem}\label{prob1}
Can we construct new classes of linear codes that are either non-GRS  MDS codes or NMDS codes, but are not derived from TGRS codes?
\end{problem}

We now turn back to the class of linear codes constructed in \cite{HZ2023}. 
From the above discussions, we can see that such codes indeed provide an affirmative answer to {\bf Problem \ref{prob1}}.  
On one hand, the self-dual properties and the properties of supporting designs were partly studied in \cite{HZ2023,HW2023,LZM2024}. 
Note also that for GRS codes, one possibly gets more self-dual codes by adding a column vector of 
the form $(0,\ldots,0,1)^T$ at the end of the generator matrices of self-dual GRS codes 
(see \cite{FXF2021,ZF2020,JX2014,JX2016,WLZ2023} and the references therein).
On the other hand, it is invariably meaningful to examine parameters, weight distributions, and self-dual properties 
of certain extended codes of a known and good family of linear codes in coding theory \cite{SD2023,HP2003,SDC2024,DTbook}. 
Therefore, {\bf the second motivation} of this paper is to fix the following problem:
\begin{problem}\label{prob2}
 What are the parameters, weight distributions, and self-dual properties of the class of linear codes 
 derived by appending $(0,\ldots,0,1)^T$ to the end of the generator matrices of the linear codes constructed in \cite{HZ2023}?
\end{problem}
}

Inspired by {\bf Problems \ref{prob1}} and {\bf \ref{prob2}}, we study a special class of linear codes $\C_k(\SSS,{\bf v},\infty)$ 
(see Definition \ref{def.codes} for more details) in this paper.  
%Note that their generator matrices $G_k$ (see Equation (\ref{eq.generator matrix})) are indeed obtained by adding a column in the generator matrices (see \cite[Equation 1.1]{HZ2023}) of linear codes constructed in \cite{HZ2023}. 
Our main contributions can be summarized as follows: 
\begin{itemize}
    \item [\rm 1)] We completely determine the parameters of $\C_k(\SSS,{\bf v},\infty)$ and prove they must be MDS codes or NMDS codes 
    in Theorems \ref{th.length and dimension} and \ref{th.must be MDS or NMDS}. 
    We also characterize sufficient and necessary conditions for $\C_k(\SSS,{\bf v},\infty)$ to be MDS and NMDS codes in Theorem \ref{th.MDS condition}. 
    Based on the Schur method, we further show that such codes are non-GRS in Theorem \ref{th.non-GRS}. 
    Note also that the way to generate $\C_k(\SSS,{\bf v},\infty)$ is different from TGRS codes and 
    this fact has been confirmed by the detailed comparisons presented in Remarks \ref{rem.111} and \ref{rem.weight distribution}. 
    {Therefore, $\C_k(\SSS,{\bf v},\infty)$ actually provides another positive answer to {\bf Problem \ref{prob1}}}.

    \item [\rm 2)] With the help of the solutions to subset sum problems, we completely determine the weight distributions 
    of $\C_k(\SSS,{\bf v},\infty)$ in Theorem \ref{th.NMDS weight distribution}, Corollaries \ref{coro.NMDS Fq*} and \ref{coro.NMDS Fq}. 
    As an application, we immediately obtain several new 
    infinite families of $5$- and $6$-dimensional NMDS codes in 
    Theorems \ref{th.5 dimensional NMDS codes} and \ref{th.6 dimensional NMDS codes}. 
    Detailed comparisons have also been made in Remark \ref{rem.weight.Heng} to show that 
    possible restrictions and weight distributions, and hence such codes, are different from known ones.

    \item [\rm 3)] Employing some identical equations over finite fields, we characterize the self-orthogonal properties  of $\C_k(\SSS,{\bf v},\infty)$ in Theorem \ref{th.so criterion}. 
    Based on this, we deduce that there are no self-dual codes in this class of linear codes in Corollary \ref{coro.no SD codes}, 
    {although many new self-dual codes were constructed in \cite{HZ2023}. 
    This fact indeed yields a clear and important distinction from GRS codes}. 
    Alternatively, we obtain two explicit constructions of almost self-dual codes in Theorems \ref{th.ASD} and \ref{th.ASD2}. 
    {In combination with 1) and 2) above, we have answered {\bf Problem \ref{prob2}}.}     
\end{itemize}

This paper is organized as follows. 
In Section \ref{sec2.preliminaries}, we recall some basic notations and auxiliary results on linear codes. 
In Section \ref{sec3}, we study the class of linear codes $\C_k(\SSS,{\bf v},\infty)$ involving their 
parameters, non-GRS properties, weight distributions, and self-orthogonal properties. 
In Section \ref{sec.concluding remarks}, we conclude this paper.

\section{Preliminaries}\label{sec2.preliminaries}

\subsection{Basic notations}\label{sec.Basic notations}
Throughout this paper, we fix the following notations, unless they are stated otherwise: 
\begin{itemize}
    \item ${\bf v}=(v_1\ldots,v_n)\in (\F_q^*)^n$ and $\SSS=\{a_1,\ldots,a_n\}$ is a subset of $\F_q$ with  
    $
        u_i=\prod_{a_j\in \SSS, 1\leq j\neq i\leq n}(a_i-a_j)^{-1}~{\rm for~each}~1\leq i\leq n.   
    $

     \item For a linear code $\C$, $\wt(\bf c)$  is the {\em Hamming weight} of a codeword ${\bf c}\in \C$ and 
     $A_i$ denotes the {\em number of codewords of weight $i$} in $\C$ for each $1\leq i\leq n$. 
     Moreover, $\dim(\C)$ denotes the dimension of $\C$ and $d(\C)$ denotes the minimum Hamming distance of $\C$.

     \item $\{A_i:~ i=0,\ldots,n\}$ denotes the {\em weight distribution} of an $[n,k]_q$ linear code $\C$ and 
     $\{A^{\perp}_i:~ i=0,\ldots,n\}$ denotes the {\em weight distribution of the dual code} of $\C$.  
     Moreover, $A(z)=1+A_1z+\ldots+A_nz^n$ and $A^{\perp}(z)=1+A^{\perp}_1z+\ldots+A^{\perp}_nz^n$ are 
     {\em weight enumerators} of $\C$ and $\C^{\perp}$, respectively.

     \item $\F_{q}[x]_k=\{f(x)=\sum_{i=0}^{k-1}f_ix^i:~f_i\in \F_q,~0\leq i\leq k-1\}$ and  
     $
        \mathcal{V}_k=\{f(x)=\sum_{i=0}^{k-2}f_ix^i+f_kx^k:~f_i\in \F_q,~0\leq i\leq k~{\rm and}~i\neq k-1\}. 
     $

    \item {$\langle \aaa_1, \ldots, \aaa_t\rangle$ denotes 
    the linear space spanned by all the vectors $\aaa_1,\ldots, \aaa_t$.}
\end{itemize}

\subsection{MDS codes and NMDS codes} 

In this subsection, we recall some basic knowledge on MDS codes and NMDS codes. 
As we mentioned before, an important class of MDS codes are the so-called generalized Reed-Solomon (GRS) codes \cite{HP2003}.  
Let $k$ and $n$ be two positive integers satisfying $1\leq k\leq n\leq q$.  
With above notations, an $[n,k,n-k+1]_{q}$ {\em generalized Reed-Solomon (GRS) code} associated with $\SSS$ and $\mathbf{v}$ is defined by  
$
        \GRS_k(\SSS,\mathbf{v}) =  \{(v_1f(a_1),\ldots,v_nf(a_n)):~ f(x)\in {\mathbb{F}_{q}[x]_{k}}\}.
$
With \cite{HP2003}, the dual code of a GRS code is still a GRS code and a generator matrix of $\GRS_k(\SSS,\mathbf{v})$ has the form of  
\begin{equation}\label{eq.GRS.generator matrix}
G_{\GRS_k}=\left(\begin{array}{cccc}
v_1 &  \ldots & v_n \\
v_1a_1  & \ldots & v_na_n \\
\vdots &  \ddots & \vdots \\
%v_1a^{k-2}_1  & \ldots & v_na^{k-2}_n \\
v_1a^{k-1}_1  & \ldots & v_na^{k-1}_n \\
\end{array}\right).
\end{equation}

By adding a column vector of the form $(0,\ldots,0,1)^T$ to Equation (\ref{eq.GRS.generator matrix}), 
we immediately get an $[n+1,k,n-k+2]_q$ MDS code, namely an extended GRS (EGRS) code $\GRS_{k}(\SSS,\mathbf{v},\infty)$, 
with the following generator matrix
\begin{equation}\label{eq.EGRS.generator matrix}
    G_{\EGRS_k}=\left(\begin{array}{ccccc}
    v_1 &  \ldots & v_n & 0\\
    v_1a_1 &  \ldots & v_na_n & 0\\
    \vdots &  \ddots & \vdots & \vdots\\
    v_1a^{k-2}_1 &  \ldots & v_na^{k-2}_n & 0\\
    v_1a^{k-1}_1 &  \ldots & v_na^{k-1}_n & 1\\
    \end{array}\right).
    \end{equation}

%By adding a column vector $(0,0,\ldots,0,1)^T$ of length $k$ on the right of the last column of the matrix presented in Equation (\ref{eq.GRS.generator matrix}), 

For any given linear code, one can also use the following lemma to determine if it is an MDS code. 

\begin{lemma}{\rm (\!\!\cite[Lemma 7.3]{B2015})}\label{lem.MDS} 
    Let $\C$ be an $[n,k]_q$ linear code with a generator matrix $G$. 
    Then $\C$ is MDS if and only if any $k$ columns of $G$ are linearly independent. 
\end{lemma}

%For $0\leq i\leq n$, let $A_i$ (resp. $A^{\perp}_i$) be the numbers of codewords  {of} weight $i$ in $\C$ (resp. $\C^{\perp}$). Let $\{A_i\mid i=0,1,\ldots,n\}$ (resp. $\{A^{\perp}_i\mid i=0,1,\ldots,n\}$) denote the weight distribution of $\C$ (resp. $\C^{\perp}$). 
It is {well known} that all weight distributions of MDS codes with the same parameters are the same, 
but this argument does not hold generally for NMDS codes. 
In \cite{DL1994}, Dodunekov and Landjev proved that the weight distributions of an NMDS code and its dual code depend on the numbers of their minimum weight codewords. 

\begin{lemma}{\rm (\!\!\cite[Corollary 4.2]{DL1994})}\label{lem.NMDS weight distribution}
	Let $\C$ be an $[n,k]_q$ NMDS code. If $1\leq s\leq k$, then 
	\begin{align*}
        A_{n-k+s}=&\binom{n}{k-s}\sum_{i=0}^{s-1}(-1)^i\binom{n-k+s}{i}(q^{s-i}-1)\\
                 &+(-1)^s\binom{k}{s}A_{n-k}.
    \end{align*} 
	If $1\leq s\leq n-k$, then 
    \begin{align*}
        A^{\perp}_{k+s}=&\binom{n}{k+s}\sum_{i=0}^{s-1}(-1)^i\binom{k+s}{i}(q^{s-i}-1)\\
                        &+(-1)^s\binom{n-k}{s}A^{\perp}_{k}.
    \end{align*}
Moreover, $A_{n-k}=A_k^{\perp}$, $A_0=A^{\perp}_0=1$ and $A_1=A_2=\ldots=A_{n-k-1}=A^{\perp}_1=A^{\perp}_2=\ldots=A^{\perp}_{k-1}=0$.
\end{lemma}

\subsection{The subset sum problem and the Schur product}

Let $\mathcal{S}\subseteq \F_q$ and $b\in \F_q$. 
The  {\em{subset sum problem}} over $\mathcal{S}$ is to determine 
if there is a subset $\emptyset \neq \{x_1,\ldots,x_{\ell}\}\subseteq \mathcal{S}$ such that 
\begin{align}\label{eq.subset sum problem def}
	x_1+\ldots+x_\ell=b.     
\end{align}
Let $N({\ell},b,\mathcal{S})$ be the number of subsets $\{x_1,\ldots,x_{\ell}\}\subseteq \mathcal{S}$ such that Equation (\ref{eq.subset sum problem def}) holds. 
For $b=0$, we {say that} $\mathcal{S}$ contains an {\em{${\ell}$-zero-sum subset}} if $N({\ell},0,\mathcal{S})>0$; 
and {say that} $\mathcal{S}$  is {\em{${\ell}$-zero-sum free}} if $N({\ell},0,\mathcal{S})=0$. 
It is usually interesting but difficult to determine the value of $N({\ell},b,\mathcal{S})$. 
In \cite{LW2008}, Li and Wan determined the exact values of $N({\ell},0,\mathcal{S})$ for $\mathcal{S}=\F_q^*$ and $\F_q$.  

\begin{lemma}{\rm (\!\!\cite[Theorem 1.2]{LW2008})}\label{lem.subset sum}
	{Let the notations be} the same as above. Then the following statements hold. 
    \begin{enumerate}
        \item [\rm 1)] If $\SSS=\F_q^*$, then $N({\ell},0,\F_q^*)=\frac{1}{q}\left[\binom{q-1}{{\ell}}+(-1)^{\ell+\lfloor \frac{{\ell}}{p} \rfloor} (q-1)\binom{\frac{q}{p}-1}{\lfloor \frac{{\ell}}{p} \rfloor} \right]$. 
        \item [\rm 2)] If $\SSS=\F_q$, then $N({\ell},0,\F_q)=\left\{ \begin{array}{ll}
            \frac{1}{q}\binom{q}{{\ell}}, & ~{\rm if}~ p\nmid \ell, \\ 
            \frac{1}{q}\left[\binom{q}{{\ell}}+(-1)^{\ell+\frac{{\ell}}{p}} (q-1)\binom{\frac{q}{p}}{\frac{{\ell}}{p}} \right], & ~{\rm if}~ p\mid \ell. 
        \end{array}\right.$
    \end{enumerate}
\end{lemma}

%\begin{definition}\label{def.schur product}
%\end{definition}

The {\em Schur product of two vectors} ${\bx}=(x_1,\ldots,x_n)$ 
and ${\bf y}=(y_1,\ldots,y_n)$ in $\F_q^n$ is defined by 
$\bx \star \by= (x_1y_1,\ldots,x_ny_n).$  
Based on this, the {\em Schur product of two linear codes} $\C_1$ and $\C_2$ of the same length is given by 
$\C_1\star \C_2=\{ {\bf c}_1\star {\bf c}_2:~{\bf c}_1\in \C_1~{\rm and}~{\bf c}_2\in \C_2 \}.$  
Moreover, if $\C_1=\langle \bc_{11},  \ldots, \bc_{1k}\rangle$ 
and $\C_2=\langle \bc_{21},  \ldots, \bc_{2\ell}\rangle$ 
are two $q$-ary linear codes with the same length, then their Schur product $\C_1\star \C_2$ is also a $q$-ary  linear code given by 
\begin{align*}
    \C_1\star \C_2=\langle {\bf c}_{1i}\star {\bf c}_{2j}:~1\leq i\leq k~{\rm and}~1\leq j\leq \ell\rangle. 
\end{align*}
If $\C_1=\C_2=\C$, we abbreviate $\C_1\star \C_2$ as $\C^2$.

Note that if $\C$ and $\D$ are two equivalent linear codes, then $\C^2$ is equivalent to $\D^2$. 
Based on the following two lemmas, we can judge if a linear code is equivalent to a GRS code.

\begin{lemma}{\rm (\!\!\cite[Proposition 10]{MMP2013},\cite[Proposition 2.1 (1)]{ZL2024})}\label{lem.GRS square dimension}
    If $\C$ is an $[n+1,k]_q$ GRS code with $3\leq k<\frac{n+2}{2}$, 
    then $\dim({\C^2})=2k-1.$ 
\end{lemma}

\begin{lemma}{\rm (\!\!\cite[Lemma 3.3]{HL2023})}\label{lem.GRS square distance}
    If $\C$ is an $[n+1,k]_q$ GRS code with $3\leq k< \frac{n+2}{2}$, 
    then $d({\C^2})\geq 2.$ 
\end{lemma}

\subsection{Some auxiliary results} 

At the end of this section, we recall some auxiliary results for later use. 

\begin{lemma}{\rm(\!\!\cite[Corollary 3.6]{LP2020})}\label{lem.so condition}
    Let $\C$ be an $[n,k]_q$ linear code with a generator matrix $G$. 
    Then $\C$ is self-orthogonal if and only if $GG^T$ is a zero matrix of size $k\times k$. 
    %Then $\dim(\Hull(\C))=k-\rank(GG^T).$ Moreover, $\C$ is self-orthogonal if and only if $GG^T$ is a zero matrix of size $k\times k$. 
\end{lemma}

\begin{lemma}{\rm(\!\!\cite[Page 466]{H1929}, \cite[Lemma 17]{LDMTT2021})}\label{lem.det}
    For any integer $n\geq 3$, it holds that 
    \begin{align*}
        \det\left(\begin{array}{ccccc}
            1 &  \ldots &  1 \\ 
            a_1  & \ldots &  a_n \\
            %a_1^2 &  \ldots & a_{n-1}^2 & a_n^2 \\
            \vdots &  \ddots &  \vdots \\
            a_1^{n-2} &  \ldots &  a_n^{n-2} \\
            a_1^n &  \ldots &  a_n^n 
        \end{array}\right)=\sum_{s=1}^{n}a_s\prod_{1\leq i<j\leq n}(a_j-a_i).             
    \end{align*}
\end{lemma}
 
%By Lemma \ref{lem.det}, we have the following result. 
 
\begin{lemma}%{\rm(\!\!\cite[Lemma 5]{LXW2008})}
    \label{lem.PRA}
    For any subset $\SSS=\{a_1,\ldots,a_n\}\subseteq \F_q$ with $n\geq 3$, it holds that  
    \begin{align}\label{eq.ui value}
        \sum_{i=1}^{n}a_i^{\ell}u_i=\left\{\begin{array}{ll}
            0, & 0\leq \ell\leq n-2, \\
            1, & \ell=n-1, \\
            \sum_{i=1}^n a_i, & \ell=n.
        \end{array} \right.
    \end{align}
\end{lemma}
\begin{proof}
    For the cases where $0\leq \ell\leq n-2$ and $\ell=n-1$, the desired results have been documented in 
    \cite[Theorem 3.1]{LXW2008} and \cite[Lemma \Rmnum{2}.1]{FLL2020}, respectively. 
    For the left case, the desired result follows immediately by combining Lemma \ref{lem.det} and the proof given in \cite[Theorem 3.1]{LXW2008}.  
\end{proof}

\section{The class of linear codes $\C_k(\SSS,{\bf v},\infty)$}\label{sec3}

\begin{definition}\label{def.codes}
    Let $n$ and $k$ be two positive integers satisfying $3\leq k\leq n-2\leq q-2$.
    Let $\SSS=\{a_1,\ldots,a_{n}\}\subseteq \F_q$ and 
    ${\bf v}=(v_1,\ldots,v_{n})\in (\F_q^*)^n$.  
    Denote by $\C_k(\SSS,{\bf v},\infty)$ the $q$-ary linear code generated by the matrix 
        \begin{equation}\label{eq.generator matrix}
            G_k=\begin{pmatrix}    
        v_1 &  \ldots & v_{n} & 0 \\ 
        v_1a_1 &  \ldots & v_{n}a_{n} & 0 \\
        \vdots &  \ddots & \vdots & \vdots \\
        v_1a_1^{k-2} &  \ldots & v_{n}a_{n}^{k-2} & 0 \\
        v_1a_1^k &  \ldots & v_{n}a_{n}^k & 1 \\
            \end{pmatrix}.
        \end{equation}        
        With the notation $\mathcal{V}_k$ given in Subsection \ref{sec.Basic notations}, $\C_k(\SSS,{\bf v},\infty)$ can be further expressed as  
        \begin{align*}
            \{ \left(v_1f(a_1),\ldots,v_nf(a_n),f_{k}\right): ~f(x)\in \mathcal{V}_k\},  
        \end{align*}
    where $f_{k}$ is the coefficient of $x^k$ in $f(x)$.     
\end{definition}
%Furthermore, if $\SSS=\F_q^*$ (resp. $\F_q$), we write $\C_k(\F_q^*,{\bf v},\infty)$ (resp. $\C_k(\F_q,{\bf v},\infty)$). 
%Hence, $\C_k(\SSS,{\bf v},\infty)$ can be seen as a subcode of the well-known extended generalized Reed-Solomon code $\GRS_{k+1}({\bf a},{\bf v},\infty)$. 

In the following, we focus on the parameters, non-GRS properties, weight distributions, 
and self-orthogonal properties of $\C_k(\SSS,{\bf v},\infty)$.  

\subsection{Parameters of $\C_k(\SSS,{\bf v},\infty)$} 

In this subsection, we determine the parameters of $\C_k(\SSS,{\bf v},\infty)$. 
We will prove that $\C_k(\SSS,{\bf v},\infty)$ must be MDS or NMDS codes 
and give necessary and sufficient conditions for them. 
%but also fully establish under what conditions they are MDS or NMDS. 
First of all, we determine the length and dimension of  $\C_k(\SSS,{\bf v},\infty)$. 

\begin{theorem}\label{th.length and dimension} 
    Let $\C_k(\SSS,{\bf v},\infty)$ be the $q$-ary linear code generated by $G_k$ given in Equation (\ref{eq.generator matrix}). 
    Then $\C_k(\SSS,{\bf v},\infty)$ has parameters $[n+1,k]_q$. 
\end{theorem}
\begin{proof} 
    It is clear that the length of $\C$ is $n+1$. 
    Denote the $i$-th row of $G_k$ by $\mathbf{g}_i$ for $1\leq i\leq k$. 
    To prove that $\dim(\C_k(\SSS,{\bf v},\infty))=k$, it suffices to show that 
    $\mathbf{g}_1,  \ldots, \mathbf{g}_k$ are linearly independent over $\F_q$. 
    Assume that there are $u_1,\ldots,u_k\in\F_q$ such that 
    $u_1\mathbf{g}_1+\ldots+u_k\mathbf{g}_k={\bf 0}$. 
    This implies that $u_k=0$ and the polynomial 
    $f(x)=u_{k-1} x^{k-2}+\ldots+u_2x+u_1$ 
    has $n$ distinct roots $a_1,\ldots,a_n$.  
    Since $n> k-2$, it follows that $u_1=\ldots=u_k=0$, and hence, 
    $\mathbf{g}_1,  \ldots, \mathbf{g}_k$ are linearly independent. 
    This completes the proof. 
\end{proof}

%Then we obtain a parity-check matrix of $\C_k(\SSS,{\bf v},\infty)$ from the form of the generator matrix $G_k$ presented in  Equation (\ref{eq.generator matrix}). 

\begin{theorem}\label{th.parity check matrix}    
    Let $\C_k(\SSS,{\bf v},\infty)$ be the $q$-ary linear code generated by $G_k$ given in Equation (\ref{eq.generator matrix}). 
    Suppose that $v'_i=u_iv_i^{-1}$ for any $1\leq i\leq n$. 
    Then the following statements hold. 
    \begin{enumerate}
        \item [\rm 1)] The $(n-k+1)\times (n+1)$ matrix 
        \begin{equation*}%\label{eq.parity check matrix}
            \footnotesize
            H_{n-k+1}=\begin{pmatrix}    
        v'_1 &  \ldots & v'_{n} & 0 \\ 
        v'_1a_1  & \ldots & v'_{n}a_{n} & 0 \\
        \vdots &  \ddots & \vdots & \vdots \\
        v'_1a_1^{n-k-2}  & \ldots & v'_{n}a_{n}^{n-k-2} & 0 \\
        v'_1a_1^{n-k-1}  & \ldots & v'_{n}a_{n}^{n-k-1} & -1 \\
        v'_1a_1^{n-k}  & \ldots & v'_{n}a_{n}^{n-k} & -\sum_{i=1}^na_i \\
            \end{pmatrix}
        \end{equation*}
        is a parity-check matrix of $\C_k(\SSS,{\bf v},\infty)$. 
    
        \item [\rm 2)] The dual code $\C_k(\SSS,{\bf v},\infty)^{\perp}$ is given by
        $\footnotesize
            \{(v_1'g(a_1), \ldots, v_n'g(a_n), -g_{n-k-1}-g_{n-k} \sum_{i=1}^na_i):~
            g(x)\in \F_q[x]_{n-k+1}\}. 
        $
        %$$\C_k(\SSS,{\bf v},\infty)^{\perp}=\left\{\left(v_1'g(a_1), v_2'g(a_2), \ldots, v_n'g(a_n), -g_{n-k-1}-\sum_{i=1}^na_i g_{n-k}\right):~g(x)\in \F_q[x]~{\rm with}~\deg(g(x))\leq n-k \right\}.$$
    \end{enumerate}
    \end{theorem}
    \begin{proof}
        1) On one hand, it is easy to see that $\rank(H_{n-k+1})=n-k+1$. 
        %On the other hand, we denote ${\bf h}_t=(v'_1a_1^{t-1}, v'_2a_2^{t-1}, \ldots, v'_{n}a_{n}^{t-1}, b)$ for $1\leq t\leq n-k+1$, 
        %where %$b=0$ if $1\leq t\leq n-k-1$, $b=-1$ if $t=n-k$ and $b=-\sum_{i=1}^{n}a_i$ if $t=n-k+1$.  
        %\begin{align*}
        %    b=\left\{
        %        \begin{array}{ll}
        %            0, & {\rm if}~1\leq t\leq n-k-1, \\
        %            -1, & {\rm if}~t=n-k, \\
        %            -\sum_{i=1}^{n}a_i, & {\rm if}~t=n-k+1, \\
        %        \end{array}
        %    \right.
        %\end{align*}
        %${\bf h}_{n-k-1}=(v'_1a_1^{n-k-1}, v'_2a_2^{n-k-1}, \ldots, v'_{n}a_{n}^{n-k-1}, -1)$ and ${\bf h}_{n-k}=(v'_1a_1^{n-k}, v'_2a_2^{n-k}, \ldots, v'_{n}a_{n}^{n-k}, -\sum_{i=1}^{n}a_i)$, where $1\leq i\leq n-k-1$. 
        %Let $H_{n+k-1}$ be an $(n-k+1)\times (n+1)$ matrix whose $t$-th row is ${\bf h}_t$ for $0\leq t\leq n-k+1$. 
        On the other hand, we denote the $s$-th row of the generator matrix $G_k$ by ${\bf g}_s$ for any $1\leq s\leq k$ 
        and denote the $t$-th row of the matrix $H_{n-k+1}$ by $\mathbf{h}_t$ for any $1\leq t\leq n-k+1$.  
        Our main task, then, becomes to prove ${\bf g}_s {\bf h}^T_t=0$ for any possible $s$ and $t$. 
        It is not difficult to calculate that  
        \begin{align*}
            \small
            {\bf g}_s {\bf h}^T_t=\left\{\begin{array}{ll}
                \sum_{i=1}^n u_ia_i^{s+t-2}, & {\rm if}\hspace{-4mm}\begin{array}{ll}&1\leq s\leq k-1~{\rm and}\\ &1\leq t\leq n-k+1,\end{array} \\ 
                \sum_{i=1}^n u_ia_i^{k+t-1}, & {\rm if}\hspace{-4mm}\begin{array}{ll}&s=k~{\rm and}\\ &1\leq t\leq n-k-1,\end{array} \\ 
                \sum_{i=1}^n u_ia_i^{n-1}-1, & {\rm if}~s=k~{\rm and}~ t=n-k, \\ 
                \sum_{i=1}^n u_ia_i^{n}-\sum_{i=1}^na_i, & {\rm if}\hspace{-4mm}\begin{array}{ll}&s=k~{\rm and}\\ &t=n-k+1.\end{array} 
            \end{array} \right.
        \end{align*}
        Then the desired result 1) follows straightforward from Lemma \ref{lem.PRA}. 
    
        2) Since a parity-check matrix of a linear code is a generator matrix of its dual code, 
        then the expected result 2) immediately holds according to the result 1) above. %This completes the proof. 
    \end{proof}

\begin{theorem}\label{th.must be MDS or NMDS} 
    Let $\C_k(\SSS,{\bf v},\infty)$ be the $q$-ary linear code generated by $G_k$ given in Equation (\ref{eq.generator matrix}). 
    Then $\C_k(\SSS,{\bf v},\infty)$ is either  MDS or NMDS. 
\end{theorem}
\begin{proof} 
    We first show that $\C_k(\SSS,{\bf v},\infty)$ is MDS or AMDS. 
    Note that the $(k+1)\times (n+1)$ matrix $G_{\EGRS_{k+1}}$ given in Equation (\ref{eq.EGRS.generator matrix}) 
    can generate an $[n+1,k+1,n-k+1]_q$ MDS code $\GRS_{k+1}(\SSS,{\bf v},\infty)$.  
    Since $G_k$ is a submatrix of $G_{\EGRS_{k+1}}$, then $\C_k(\SSS,{\bf v},\infty)$ is a subcode of $\GRS_{k+1}(\SSS,{\bf v},\infty)$. 
    It implies that $d(\C_k(\SSS,{\bf v},\infty))\geq n-k+1$. 
    On the other hand, $d(\C_k(\SSS,{\bf v},\infty))\leq n-k+2$ according to Theorem \ref{th.length and dimension} and the well-known Singleton bound. 
    Combining these two aspects, we conclude that $d(\C_k(\SSS,{\bf v},\infty))\in \{n-k+1, n-k+2\}$. 
    Hence, $\C_k(\SSS,{\bf v},\infty)$ is MDS or AMDS. 

    Next, we prove that $\C_k(\SSS,{\bf v},\infty)^{\perp}$ is also MDS or AMDS. 
    Let ${\bf c}^{\perp}$ be any nonzero codeword in $\C_k(\SSS,{\bf v},\infty)^{\perp}$. 
    By Theorem \ref{th.parity check matrix} 2), there exists a polynomial $g(x)\in \F_q[x]_{n-k+1}$ such that 
    ${\bf c}^{\perp}=(v_1'g(a_1), \ldots, v_n'g(a_n), -g_{n-k-1}-g_{n-k}\sum_{i=1}^na_i).$
    Since $\deg(g(x))\leq n-k$, we get $$\wt({\bf c}^{\perp})\geq n-(n-k)=k.$$ 
    It turns out that $d(\C_k(\SSS,{\bf v},\infty)^{\perp})\geq k$. 
    Again, according to the Singleton bound and similar arguments above, we know that $\C_k(\SSS,{\bf v},\infty)^{\perp}$ is MDS or AMDS. 
    Therefore, the expected result follows from the fact that the dual code of an MDS code is still MDS.  
    This completes the proof. 
\end{proof}

%\begin{theorem}\label{th.d bound} 
%    Let $\C_k(\SSS,{\bf v},\infty)$ be the $q$-ary linear code generated by $G_k$ given in Equation (\ref{eq.generator matrix}). 
%    Then $\C_k(\SSS,{\bf v},\infty)$ is MDS or AMDS. 
%\end{theorem}
%\begin{proof} 
%    Since $\GRS_{k+1}({\bf a},{\bf v},\infty)$ is an $[n+1,k+1,n-k+1]_q$ MDS code for given $\bf{a}$ and ${\bf v}$, 
%    it follows from $\C_k(\SSS,{\bf v},\infty)\subseteq \GRS_{k+1}({\bf a},{\bf v},\infty)$ that $d(\C_k(\SSS,{\bf v},\infty))\geq n-k+1$. 
%    On the other hand, $d(\C_k(\SSS,{\bf v},\infty))\leq n-k+2$ according to the well-known Singleton bound. 
%    By definitions, we have that $\C_k(\SSS,{\bf v},\infty)$ is AMDS if $d(\C_k(\SSS,{\bf v},\infty))=n-k+1$ and 
%    MDS if $d(\C_k(\SSS,{\bf v},\infty))=n-k+2$. This completes the proof. 
%\end{proof}

%According to Theorem \ref{th.length and dimension}, $\C_k(\SSS,{\bf v},\infty)$ is an $[n+1,k,n-k+1]_q$ AMDS code or an $[n+1,k,n-k+2]_q$ MDS code. 
Motivated by Theorem \ref{th.must be MDS or NMDS}, we can further explore under what conditions $\C_k(\SSS,{\bf v},\infty)$ is MDS or NMDS.  
A complete characterization based on certain zero-sum conditions is obtained in the following theorem. 

\begin{theorem}\label{th.MDS condition}
%Let $n$ and $k$ be positive integers such that $n\geq k-1$. 
Let $\SSS=\{a_1,\ldots,a_{n}\}\subseteq \F_q$ and $\C_k(\SSS,{\bf v},\infty)$ be the $q$-ary linear code 
generated by $G_k$ given in Equation (\ref{eq.generator matrix}). 
Then the following statements hold. 
\begin{enumerate}
    \item [\rm 1)] $\C_k(\SSS,{\bf v},\infty)$ is an $[n+1,k,n-k+2]_q$ MDS code if and only if $\SSS$ is $k$-zero-sum free. 
    \item [\rm 2)] $\C_k(\SSS,{\bf v},\infty)$ is an $[n+1,k,n-k+1]_q$ NMDS code if and only if $\SSS$ contains a $k$-zero-sum subset. 
\end{enumerate}
\end{theorem}
\begin{proof}
    1) 
    Let $\mathbf{c}_{i_1},  \ldots, \mathbf{c}_{i_k}$ be any $k$ columns of the matrix $G_k$.  
    Then $\mathbf{c}_{i_1},  \ldots, \mathbf{c}_{i_k}$ can form a $k\times k$ matrix of the form 
$$
M_{k\times k}^{(1)}=
\begin{pmatrix}
    v_{i_1}  & \ldots & v_{i_k} \\ 
    v_{i_1}a_{i_1}  & \ldots & v_{i_k} a_{i_k}\\
    \vdots &  \ddots & \vdots\\
    v_{i_1}a_{i_1}^{k-2}  & \ldots & v_{i_k} a_{i_k}^{k-2}\\
    v_{i_1}a_{i_1}^k  & \ldots & v_{i_k} a_{i_k}^k\\
\end{pmatrix}
~{\rm or}
$$
$$
M_{k\times k}^{(2)}=
\begin{pmatrix}
    v_{i_1} &  \ldots & v_{i_{k-1}} & 0 \\ 
    v_{i_1}a_{i_1}  & \ldots & v_{i_{k-1}} a_{i_{k-1}} & 0 \\
    \vdots &  \ddots & \vdots & \vdots \\
    v_{i_1}a_{i_1}^{k-2}  & \ldots & v_{i_{k-1}} a_{i_{k-1}}^{k-2} & 0 \\
    v_{i_1}a_{i_1}^k &  \ldots & v_{i_{k-1}} a_{i_{k-1}}^k & 1 \\
\end{pmatrix}. 
$$
It follows from Lemma \ref{lem.det} that 
$$\det \left(M_{k\times k}^{(1)}\right)=(a_{i_1}+\ldots+a_{i_k})\prod_{j=1}^{k}v_{i_j} \prod_{1\le s<t\le k}(a_{i_t}-a_{i_s})$$ 
$${\rm and}~\det \left(M_{k\times k}^{(2)}\right)=\prod_{j=1}^{k-1}v_{i_j} \prod_{1\le s<t\le k-1}(a_{i_t}-a_{i_s}).$$
On one hand, $\det \left(M_{k\times k}^{(2)}\right)\neq 0$ since $v_{i_j}\neq 0$ and $a_{i_s}\neq a_{i_t}$ for any $1\leq j \leq k$ and $1\leq s\neq t\leq k$.  
On the other hand, we have that $\det \left(M_{k\times k}^{(1)}\right)\neq 0$ if and only if $a_{i_1}+\ldots+a_{i_k}\neq 0$. 
Then the desired result 1) follows from Lemma \ref{lem.MDS} and Theorem \ref{th.length and dimension}. 

2) Combining Theorem \ref{th.must be MDS or NMDS} and the result 1) above, we immediately get the expected result 2). 
This completes the proof. 
\end{proof}

\begin{remark}\label{rem.111} 
    Note that Zhu and Liao recently studied in \cite{ZL2024} some properties of the so-called (+)-extended TGRS code 
    $\C_{1,k-1,k,n}(\SSS, {\bf v}, \eta, \infty)$, where $\eta\in \F_q^*$. 
    According to \cite[Definitions 2.2 and 2.3]{ZL2024} and \cite[Remark 2.1]{ZL2024}, 
    it can be seen that $\C_k(\SSS,{\bf v},\infty)$ and $\C_{1,k-1,k,n}({\bf a}, {\bf v}, \eta, \infty)$ are different and not equivalent in general. 
    It should be emphasized that lots of concrete examples and comparisons can be found in Remark \ref{rem.weight distribution} below. 
\end{remark}

\begin{remark}\label{rem.1} $\quad$
    \begin{itemize}
        \item [\rm 1)]   
        In a natural way, one can further consider the linear code $\C_{k}(\SSS,{\bf v},\infty)_\mu$ 
        generated by the matrix obtaining by deleting the $\mu$-th row of the generator matrix $G_{\EGRS_{k+1}}$ 
        given in Equation (\ref{eq.EGRS.generator matrix}), where $1\leq \mu\leq k-1$. 
        Note also that it is trivial to delete the last row, $i.e.,$ the $(k+1)$-th row of $G_{\EGRS_{k+1}}$. 
        Similarly to Theorem \ref{th.length and dimension}, 
        we can deduce that $\C_{k}(\SSS,{\bf v},\infty)_\mu$ has parameters $[n+1,k]_q$ for any $1\leq \mu\leq k-1$. 
        Furthermore, for any $1\leq \mu\leq k-1$, 
        using Lemma \ref{lem.MDS} again, we can obtain that 
        $\C_{k}(\SSS,{\bf v},\infty)_{\mu}$ is an $[n+1,k,n-k+2]_q$ MDS code 
        if and only if 
        \begin{align*}
            & \sum_{I\subseteq \SSS_k,~|I|=k-\mu+1}\prod_{a\in I}a \neq 0~{\text{and}} \\
            & \sum_{T\subseteq \SSS_{k-1},~|T|=k-\mu}\prod_{b\in T}b \neq 0 
        \end{align*}
        hold simultaneously for any subsets $\SSS_k=\{a_{i_1},\ldots,a_{i_k}\}\subseteq \SSS~{\text{and}}~\SSS_{k-1}=\{b_{i_1},\ldots,b_{i_{k-1}}\}\subseteq \SSS.$
        However, other results of $\C_k(\SSS,{\bf v},\infty)$ obtained in this subsection can not be easily generalized to $\C_{k}(\SSS,{\bf v},\infty)_\mu$. 

        \item [\rm 2)]  In fact, we even can not obtain new NMDS codes from $\C_{k}(\SSS,{\bf v},\infty)_\mu$ in many cases. 
        After plenty of Magma \cite{Magma} experiments, we find that most of them are not NMDS codes and 
        others are equivalent to $\C_k(\SSS,{\bf v},\infty)$. 
        Verified by Magma \cite{Magma}, two specific examples are listed in Table \ref{tab:newadd1}, 
        where the notation ``$\cong$'' indicates that two linear codes are equivalent.
        Additionally, if we delete more rows, then the situations will be more complicated 
        and numerical results similar to Table \ref{tab:newadd1} will happen more frequently. 
        In summary, these cases are beyond the scope of {\bf Problems \ref{prob1}} and {\bf \ref{prob2}},  
        and hence, we do not discuss them in detail below.       
    \end{itemize}
\end{remark}

\begin{table}[h!]
    \centering
    \caption{{Some specific examples of $\C_{k}(\SSS,{\bf v},\infty)_\mu$ discussed in Remark \ref{rem.1}}}\label{tab:newadd1}       % Give a unique label
    \vspace{-2mm}
    \scalebox{0.8}{\begin{tabular}{cccl}
     \hline 
     \multicolumn{4}{l}{$\SSS=\F_{13}\setminus \{4,6,10,11\}$,~$k=4$,~{\text{and any}}~${\bf v}\in (\F_{13}^*)^{9}$}        \\\hline \hline 
    $\mu$ & $\C_{4}(\SSS,{\bf v},\infty)_\mu$ & $(\C_{4}(\SSS,{\bf v},\infty)_\mu)^{\perp}$ & New NMDS codes? \\ \hline  
    1 & $[10,4,6]_{13}$ & $[10,6,1]_{13}$ & No (not NMDS) \\
    2 & $[10,4,6]_{13}$ & $[10,6,4]_{13}$ & No ($\C_{4}(\SSS,{\bf v},\infty)_2\cong  \C_4(\SSS,{\bf v},\infty)$)\\
    3 & $[10,4,6]_{13}$ & $[10,6,3]_{13}$ & No (not NMDS)\\\hline

    \multicolumn{4}{l}{$\SSS'=\F_{17}\setminus \{8,16\}$,~$k=5$,~{\text{and any}}~${\bf v}'\in (\F_{17}^*)^{15}$}\\\hline \hline 
    $\mu$ & $\C_{5}(\SSS',{\bf v}',\infty)_\mu$ & $(\C_{5}(\SSS',{\bf v}',\infty)_\mu)^{\perp}$ & New NMDS codes? \\ \hline  
    1 & $[16,5,11]_{17}$ & $[16,11,1]_{17}$ & No (not NMDS) \\
    2 & $[16,5,11]_{17}$ & $[16,11,5]_{17}$ & No ($\C_{5}(\SSS',{\bf v}',\infty)_2\cong \C_5(\SSS',{\bf v}',\infty)$)\\
    3 & $[16,5,11]_{17}$ & $[16,12,4]_{17}$ & No (not NMDS)\\ 
    4 & $[16,5,11]_{17}$ & $[16,12,4]_{17}$ & No (not NMDS)\\\hline
    \end{tabular}}
\end{table}

\subsection{Non-GRS properties of $\C_k(\SSS,{\bf v},\infty)$}

In this subsection, we determine the non-GRS properties of $\C_k(\SSS,{\bf v},\infty)$.  
Up to equivalence, it suffices to consider $\C_k(\SSS, {\bf 1},\infty)$, where ${\bf 1}$ is a all-one vector. 
To this end, we focus on the Schur products ${\C_k(\SSS, {\bf 1},\infty)^2}$ and 
${(\C_k(\SSS,{\bf 1},\infty)^{\perp})^2}$.

\begin{theorem}\label{th.square C}
    Let ${\C_k(\SSS,{\bf 1},\infty)}$ be the $q$-ary linear code generated by $G_k$ given in Equation (\ref{eq.generator matrix}). 
    Then the following statements hold. 
    \begin{enumerate}
        \item [\rm 1)] If $3\leq k< \frac{n}{2}$, then 
            $${\C_k(\SSS, {\bf 1},\infty)^2}=\C_{2k}(\SSS, {\bf 1}, \infty).$$

        \item [\rm 2)] If $\frac{n+1}{2}\leq k\leq n-2$, then 
            $${\C_k(\SSS, {\bf 1},\infty)^2}=\F_q^{n+1}.$$

        \item [\rm 3)] If $n$ is even and $n=2k$, then 
            $${\C_k(\SSS, {\bf 1},\infty)^2}=\{(g(a_1),\ldots,g(a_n),g_{2k}): ~g(x)\in \mathcal{V}_{2k} \}.$$

    \end{enumerate} 
\end{theorem}
\begin{proof} 
    Denote ${\bf a}^i=(a_1^i,\ldots,a_n^i)$ and $({\bf a}^i, b)=(a_1^i,\ldots,a_n^i,b)$ for any positive integer $i$ and $b\in \{0,1\}$. 
    By definitions, we have 
    \begin{align}\label{eq. square}
        \small
        \begin{split}
            & {\C_k(\SSS, {\bf 1},\infty)^2} \\ 
            = & \langle ({\bf a}^i,0)\star ({\bf a}^j,0), ({\bf a}^i,0)\star ({\bf a}^k,1), ({\bf a}^k,1)\star ({\bf a}^j,0),\\ 
              & ~~({\bf a}^k,1)\star ({\bf a}^k,1):~0\leq i,j\leq k-2 \rangle \\ 
            = & \langle ({\bf a}^{i+j},0), ({\bf a}^{i+k},0), ({\bf a}^{k+j},0), ({\bf a}^{2k},1):\\
             &~~0\leq i,j\leq k-2 \rangle \\ 
            = & \langle ({\bf a}^{u},0), ({\bf a}^{2k},1):~0\leq u\leq 2k-2 \rangle.  
        \end{split}
    \end{align}

    1) If  $3\leq k<\frac{n}{2}$, then $3< 2k< n$.   
    From Equation (\ref{eq. square}) and Definition \ref{def.codes}, we can further conclude that 
    $
             {\C_k(\SSS, {\bf 1},\infty)^2}  
          = \{(f(a_1),\ldots,f(a_n),f_{2k}):~f(x)\in \mathcal{V}_{2k}\}  
          = \C_{2k}(\SSS, {\bf 1}, \infty). 
    $
    This completes the proof of the {expected} result 1). 

    2) If  $\frac{n+1}{2}\leq k\leq n-2$, then $2k\geq n+1$.   
    On one hand, it is clear that  ${\C_k(\SSS, {\bf 1},\infty)^2}\subseteq \F_q^{n+1}.$ 
    On the other hand, since $2k\geq n+1$, then $2k-2\geq n-1$, which implies that 
    $
        \F_q^n\times \{0\}=\langle ({\bf a}^{\ell},0):~0\leq \ell\leq n-1\rangle 
                          \subseteq \langle ({\bf a}^{u},0):~0\leq u\leq 2k-2 \rangle. 
    $
    Moreover, we have  
    $\F_q^{n+1}=\langle \F_q^n\times \{0\}, ({\bf a}^{2k},1)\rangle
    \subseteq \langle ({\bf a}^{u},0), ({\bf a}^{2k},1):~0\leq u\leq 2k-2 \rangle=\C_k(\SSS, {\bf 1},\infty)^2.$ 
    Combining these two aspects, the desired result 2) holds. 

    3) If $n$ is even and $n=2k$, 
    it follows from Equation (\ref{eq. square}) that the result 3) holds directly.

    We have finished the whole proof. 
\end{proof}

\begin{theorem}\label{th.square dual C}
    Let ${\C_k(\SSS,{\bf 1},\infty)}$ be the $q$-ary linear code generated by $G_k$ given in Equation (\ref{eq.generator matrix}). 
    Then the following statements hold. 
    \begin{enumerate}
        \item [\rm 1)] If $3\leq k\leq \frac{n+1}{2}$, then 
            $${(\C_k(\SSS,{\bf 1},\infty)^{\perp})^2}=\F_q^{n+1}.$$

        \item [\rm 2)] If $\frac{n+1}{2}<k\leq n-2$, then $${(\C_k(\SSS,{\bf 1},\infty)^{\perp})^2}=\C_1+\C_{2n-2k+1},$$ 
        where  $\C_1=\langle (0,\ldots,0,1) \rangle$ and 
        $\C_{2n-2k+1}=
        \{( u_1^2g(a_1),\ldots,u_n^2g(a_n),  0 ):~g(x)\in \F_q[x]_{2n-2k+1}\}.$
    \end{enumerate} 
\end{theorem}
\begin{proof} 
    Denote ${\bf u}=(u_1,\ldots,u_n)$, ${\bf h}^r=(u_1a_1^r, \ldots,u_na_n^r)$ and 
    $({\bf h}^r, e)=(u_1a_1^r, \ldots,u_na_n^r,e)$ for any positive integer $r$ and $e\in \{0,-1,-\sum_{i=0}^{n}a_i\}$. 
    By definitions and Theorem \ref{th.parity check matrix}, we have 
    \begin{align}\label{eq. square dual}
        \small
        \begin{split}
            & {(\C_k(\SSS,{\bf 1},\infty)^{\perp})^2} \\ 
            = &  \langle ({\bf h}^r,0)\star ({\bf h}^w,e),~ ({\bf h}^{n-k-1},-1)\star ({\bf h}^{n-k-1},-1),~\\ 
              & ~~({\bf h}^{n-k-1},-1)\star ({\bf h}^{n-k},-\sum_{i=1}^{n}a_i), ({\bf h}^{n-k},-\sum_{i=1}^{n}a_i)\star \\
              & ~~({\bf h}^{n-k},-\sum_{i=1}^{n}a_i):~0\leq r\leq n-k-2,\\
              & ~~0\leq w\leq n-k~{\rm and}~e\in \{0,-1,-\sum_{i=0}^{n}a_i\} \rangle  \\
            = & \langle ({\bf u\star h}^{z},0), ({\bf u\star h}^{2n-2k-2},1), ({\bf u\star h}^{2n-2k-1},\sum_{i=1}^{n}a_i), \\ 
              & ~~({\bf u\star h}^{2n-2k},(\sum_{i=1}^{n}a_i)^2):~0\leq z\leq 2n-2k-2 \rangle \\ 
            = & \langle ({\bf u\star h}^{z},0), (0,\ldots,0,1), ({\bf u\star h}^{2n-2k-1},\sum_{i=1}^{n}a_i), \\
              & ~~({\bf u\star h}^{2n-2k},(\sum_{i=1}^{n}a_i)^2):~0\leq z\leq 2n-2k-2 \rangle \\ 
            = & \langle ({\bf u\star h}^{z},0), (0,\ldots,0,1), ({\bf u\star h}^{2n-2k-1},0), \\
              & ~~({\bf u\star h}^{2n-2k},0):~0\leq z\leq 2n-2k-2 \rangle\\ 
            = & \langle ({\bf u\star h}^{z},0):~0\leq z\leq 2n-2k \rangle + \langle (0,\ldots,0,1) \rangle.    
        \end{split}
    \end{align}

    1) If  $3\leq k\leq \frac{n+1}{2}$, then $2n-2k\geq n-1$.   
    By taking a very similar argument to the proof of Theorem \ref{th.square C} 2), 
    we immediately have ${(\C_k(\SSS,{\bf 1},\infty)^{\perp})^2}=\F_q^{n+1}.$ 

    2) If  $\frac{n+1}{2}<k\leq n-2$, then $4\leq 2n-2k\leq n-2<n-1$.    
    By Equation (\ref{eq. square dual}), the desired result 2) obviously holds. This completes the proof. 
\end{proof}

Based on Theorems \ref{th.square C} and \ref{th.square dual C}, we now can completely determine the non-GRS properties of $\C_k(\SSS, {\bf v},\infty)$ as follows. 

\begin{theorem}\label{th.non-GRS}
    Let $\C_k(\SSS,{\bf v},\infty)$ be the $q$-ary linear code generated by $G_k$ given in Equation (\ref{eq.generator matrix}). 
    %Let $\SSS=\{a_1,\ldots,a_{n}\}\subseteq \F_q$. 
    Then $\C_k(\SSS,{\bf v},\infty)$ is non-GRS. 
\end{theorem}
\begin{proof}
    Up to equivalence, it is sufficient to consider the case where ${\bf v}={\bf 1}$. 
    The proof can be divided into two parts. 

    {\textbf{Case 1:}} $\SSS$ contains a $k$-zero-sum subset. It follows from Theorem \ref{th.MDS condition} 2) that $\C_k(\SSS, {\bf 1},\infty)$ is an NMDS code. 
    Then $\C_k(\SSS, {\bf 1},\infty)$ is clearly non-GRS. 
    
    {\textbf{Case 2:}} $\SSS$ is $k$-zero-sum-free. It follows from Theorem \ref{th.MDS condition} 1) that $\C_k(\SSS,{\bf v},\infty)$ is an MDS code.  
    Since $3\leq k\leq n-2$, we have two subcases. 
    \begin{itemize}
        \item []{\textbf{Subcase 2.1:}} $3\leq k<\frac{n+1}{2}$. 
        Combining Theorems \ref{th.length and dimension} and \ref{th.square C}, 
        we know that $\dim(\C_k(\SSS, {\bf 1},\infty)^2)=2k> 2k-1.$ 
        By Lemma \ref{lem.GRS square dimension}, $\C_k(\SSS, {\bf 1},\infty)$ is non-GRS. 

        \item []{\textbf{Subcase 2.2:}} $\frac{n+1}{2}\leq k\leq n-2$. 
        Then $3\leq n+1-k\leq \frac{n+1}{2}<\frac{n+2}{2}$.
        According to Theorem \ref{th.square dual C}, 
        we know that $d((\C_k(\SSS, {\bf 1},\infty)^{\perp})^2)=1<2.$ 
        It turns out from Lemma \ref{lem.GRS square distance} that $\C_k(\SSS, {\bf 1},\infty)^{\perp}$ is non-GRS. 
        Recall that the dual code of a GRS code is still a GRS code. Therefore, $\C_k(\SSS, {\bf 1},\infty)$ is non-GRS. 
    \end{itemize}
    Combining {\textbf{Cases 1}} and {\textbf{2}}, we complete the proof.  
\end{proof}

\begin{remark}\label{rem.1112} 
    Combining with Theorems \ref{th.must be MDS or NMDS} and \ref{th.non-GRS} as well as Remark \ref{rem.111}, 
    we can conclude that $\C_{k}(\SSS,{\bf v},\infty)$ is either a non-GRS MDS code or an NMDS code, 
    and it is different from TGRS codes in general. 
    In other words, $\C_{k}(\SSS,{\bf v},\infty)$ has provided an affirmative answer to {\bf Problem \ref{prob1}}. 
\end{remark}

\subsection{Weight distributions of $\C_k(\SSS,{\bf v},\infty)$}

In this subsection, we calculate the weight distributions of $\C_k(\SSS,{\bf v},\infty)$. 
By Theorem \ref{th.must be MDS or NMDS}, we know that $\C_k(\SSS,{\bf v},\infty)$ must be MDS or NMDS. 
Note that all MDS codes have the same and known weight distribution for the same parameters \cite{HP2003}, 
while NMDS codes may have different weight distributions even if their parameters are the same \cite{defect,DL1994}. 
For this reason, we only focus on the weight distributions of NMDS codes from the class of linear codes $\C_k(\SSS,{\bf v},\infty)$. 

\begin{theorem}\label{th.NMDS weight distribution}
    Let $\SSS=\{a_1,\ldots,a_{n}\}\subseteq \F_q$ and $\C_k(\SSS,{\bf v},\infty)$ be the $q$-ary linear code 
    generated by $G_k$ given in Equation (\ref{eq.generator matrix}). 
        If $\SSS$ contains a $k$-zero-sum subset, %then $\C_k(\SSS,{\bf v},\infty)$ is an $[n+1,k,n-k+1]_q$ NMDS code,  
    then the weight distribution of $\C_k(\SSS,{\bf v},\infty)$ is 
    given by Equation (\ref{eq.weight.1}) 
    % The previous equation was number seven.
    % Account for the double column equations here.
    \addtocounter{equation}{1}  
    and the weight distribution of $\C_k(\SSS,{\bf v},\infty)^{\perp}$ is 
    given by Equation (\ref{eq.weight.2}). 
        % The previous equation was number eight.
    % Account for the double column equations here.
    \addtocounter{equation}{1}  
\end{theorem}
\begin{figure*}[!t]
	% ensure that we have normalsize text
	\normalsize
	% Store the current equation number.
	\setcounter{MYtempeqncnt}{\value{equation}}
	% Set the equation number to one less than the one
	% desired for the first equation here.
	% The value here will have to changed if equations
	% are added or removed prior to the place these
	% equations are referenced in the main text.
	\setcounter{equation}{7}
    \begin{align}\label{eq.weight.1}
        \small
        A_{i}=\left\{\begin{array}{ll}
            1, & {\rm if}~i=0, \\ 
            0, & {\rm if}~i\in \{1,2,\ldots,n-k\}, \\  
            (q-1)N(k,0,\SSS), & {\rm if}~i=n-k+1, \\ 
            \binom{n+1}{k-s}\sum_{j=0}^{s-1}(-1)^j \binom{n-k+1+s}{j}(q^{s-j}-1)+(-1)^s\binom{k}{s}A_{n-k+1}, & {\rm if}~i=n-k+s+1,~s\in \{1,2,\ldots,k\}. 
        \end{array} \right.
    \end{align}
	% Restore the current equation number.
	\setcounter{equation}{\value{MYtempeqncnt}}
	% The IEEE uses as a separator
	\hrulefill
	% The spacer can be tweaked to stop underfull vboxes.
	\vspace*{4pt}
\end{figure*}

\vspace{-4mm}
\begin{figure*}[!t]
	% ensure that we have normalsize text
	\normalsize
	% Store the current equation number.
	\setcounter{MYtempeqncnt}{\value{equation}}
	% Set the equation number to one less than the one
	% desired for the first equation here.
	% The value here will have to changed if equations
	% are added or removed prior to the place these
	% equations are referenced in the main text.
	\setcounter{equation}{8}
    \begin{align}\label{eq.weight.2}
        \small
        A_{i}^{\perp}=\left\{\begin{array}{ll}
            1, & {\rm if}~i=0, \\ 
            0, & {\rm if}~i\in \{1,2,\ldots,k-1\}, \\  
            (q-1)N(k,0,\SSS), & {\rm if}~i=k, \\ 
            \binom{n+1}{k+s}\sum_{j=0}^{s-1}(-1)^j \binom{k+s}{j}(q^{s-j}-1)+(-1)^s\binom{k}{s}A_{k}^{\perp}, & {\rm if}~i=k+s,~s\in \{1,2,\ldots,n-k+1\}. 
        \end{array} \right.
    \end{align}
	% Restore the current equation number.
	\setcounter{equation}{\value{MYtempeqncnt}}
	% The IEEE uses as a separator
	\hrulefill
	% The spacer can be tweaked to stop underfull vboxes.
	\vspace*{4pt}
\end{figure*}

\begin{proof}
    By Theorem \ref{th.MDS condition} 2), $\C_k(\SSS,{\bf v},\infty)$ is an $[n+1,k,n-k+1]_q$ NMDS code. 
    Firstly, we determine the number of codewords in $\C_k(\SSS,{\bf v},\infty)$ with the minimum weight $n-k+1$, 
    $i.e.$, the value of $A_{n-k+1}$. 
    Let ${\bf c}$ be any codeword in $\C_k(\SSS,{\bf v},\infty)$. 
    By Definition \ref{def.codes}, there exists a polynomial $f(x)=\sum_{i=0}^{k-2}f_ix^i+f_kx^k\in \mathcal{V}_k$ such that 
    $${\bf c}=(v_1f(a_1),\ldots,v_nf(a_n),f_k).$$ 
    By a direct observation, it is not difficult to deduce that 
    ${\wt({\bf c})}=n-k+1~{\rm if~and~only~if}~\sharp\{a\in \SSS:~f(a)=0\}=k~{\rm and}~f_k\neq 0,$
    which implies that 
    \begin{align*}
        \small
        & A_{n-k+1} \\
    = & \sharp\{\mathcal{V}_k \cap \{f(x)=\lambda \prod_{a\in \SSS'}(x-a):\\
     &~~~~~~~~~~\lambda\in \F_q^*,~\SSS'\subseteq \SSS~{\rm and}~\sharp \SSS'=k \} \} \\
    = & \sharp\{\SSS': -\lambda \sum_{a\in \SSS'}a=0,~\lambda\in \F_q^*,~\SSS'\subseteq \SSS~{\rm and}~\sharp \SSS'=k\} \\ 
    = &  (q-1)\sharp\{\SSS':\sum_{a\in \SSS'}a=0,~\SSS'\subseteq \SSS~{\rm and}~\sharp \SSS'=k\} \\
    = &  (q-1)N(k,0,\SSS),  
    \end{align*}
    where the second equation holds since 
    \begin{footnotesize}
    \begin{equation*}
        \prod_{a\in \SSS'}(x-a)=x^k-\sum_{a\in \SSS'}ax^{k-1}+ \sum_{i=0}^{k-2}(-1)^{k-i} \sum_{T\subseteq \SSS', \sharp T=k-i} \prod_{a\in T}ax^i.
    \end{equation*} 
    \end{footnotesize}

    Then the desired result follows straightforward from Lemma \ref{lem.NMDS weight distribution}. 
\end{proof}

Furthermore, we have the following explicit results over $\F_q^*$ and $\F_q$ according to Lemma \ref{lem.subset sum}. 

\begin{corollary}\label{coro.NMDS Fq*}
    Let $\C_k(\F_q^*,{\bf v},\infty)$ be the $q$-ary linear code generated by $G_k$ given in Equation (\ref{eq.generator matrix}). 
    If $k\neq q-2$ and $(p,k)\neq (2,q-3)$, then the weight distribution of $\C_k(\F_q^*,{\bf v},\infty)$ is 
    given by Equation (\ref{eq.weight.3}) 
    % The previous equation was number nine.
    % Account for the double column equations here.
    \addtocounter{equation}{1}  
    and the weight distribution of $\C_k(\F_q^*,{\bf v},\infty)^{\perp}$ is 
    given by Equation (\ref{eq.weight.4}). 
    % The previous equation was number ten.
    % Account for the double column equations here.
    \addtocounter{equation}{1}  
\end{corollary}

\begin{figure*}[!t]
	% ensure that we have normalsize text
	\normalsize
	% Store the current equation number.
	\setcounter{MYtempeqncnt}{\value{equation}}
	% Set the equation number to one less than the one
	% desired for the first equation here.
	% The value here will have to changed if equations
	% are added or removed prior to the place these
	% equations are referenced in the main text.
	\setcounter{equation}{9}
    \begin{align}\label{eq.weight.3}
        \small
        A_{i}=\left\{\begin{array}{ll}
            1, & {\rm if}~i=0, \\
            0, & {\rm if}~i\in \{1,2,\ldots,q-k-1\}, \\ 
            \frac{q-1}{q}\left[\binom{q-1}{k}+(-1)^{k+\lfloor \frac{k}{p} \rfloor}(q-1)\binom{\frac{q}{p}-1}{\lfloor \frac{k}{p} \rfloor} \right], & {\rm if}~i=q-k, \\ 
            \binom{q}{k-s}\sum_{j=0}^{s-1}(-1)^j \binom{q-k+s}{j}(q^{s-j}-1)+(-1)^s \binom{k}{s}A_{q-k}, & {\rm if}~i=q-k+s,~s\in \{1,2,\ldots,k\}. 
        \end{array} \right.
    \end{align}
	% Restore the current equation number.
	\setcounter{equation}{\value{MYtempeqncnt}}
	% The IEEE uses as a separator
	\hrulefill
	% The spacer can be tweaked to stop underfull vboxes.
	\vspace*{4pt}
\end{figure*}

\begin{figure*}[!t]
	% ensure that we have normalsize text
	\normalsize
	% Store the current equation number.
	\setcounter{MYtempeqncnt}{\value{equation}}
	% Set the equation number to one less than the one
	% desired for the first equation here.
	% The value here will have to changed if equations
	% are added or removed prior to the place these
	% equations are referenced in the main text.
	\setcounter{equation}{10}
    \begin{align}\label{eq.weight.4}
        \small
        A_{i}^{\perp}=\left\{\begin{array}{ll}
            1, & {\rm if}~i=0, \\
            0, & {\rm if}~i\in \{1,2,\ldots,k-1\}, \\ 
            \frac{q-1}{q}\left[\binom{q-1}{k}+(-1)^{k+\lfloor \frac{k}{p} \rfloor}(q-1)\binom{\frac{q}{p}-1}{\lfloor \frac{k}{p} \rfloor} \right], & {\rm if}~i=k, \\ 
            \binom{q}{k+s}\sum_{j=0}^{s-1}(-1)^j \binom{k+s}{j}(q^{s-j}-1)+(-1)^s\binom{k}{s}A_{k}^{\perp}, & {\rm if}~i=k+s,~s\in \{1,2,\ldots,q-k\}. 
        \end{array} \right.
    \end{align}
	% Restore the current equation number.
	\setcounter{equation}{\value{MYtempeqncnt}}
	% The IEEE uses as a separator
	\hrulefill
	% The spacer can be tweaked to stop underfull vboxes.
	\vspace*{4pt}
\end{figure*}

\begin{proof}
    Since $3\leq k\leq q-2$, it is easily verified that $N(k,0,\F_q^*)>0$ if $k\neq q-2$ and $(p,k)\neq (2,q-3)$, which implies that $\SSS$ contains a  $k$-zero-sum subset. 
    Combining Lemma \ref{lem.subset sum} 1) and Theorem \ref{th.NMDS weight distribution}, the desired result holds. 
    This completes the proof. 
\end{proof}

%Similar to Corollary \ref{coro.NMDS Fq*}, we have another corollary. 

\begin{corollary}\label{coro.NMDS Fq}
    Let $\C_k(\F_q,{\bf v},\infty)$ be the $q$-ary linear code generated by $G_k$ given in Equation (\ref{eq.generator matrix}). 
    If $(p,k)\neq (2,q-2)$, then the weight distribution of $\C_k(\F_q,{\bf v},\infty)$ is 
    given by Equation (\ref{eq.weight.5}) 
        % The previous equation was number 11.
    % Account for the double column equations here.
    \addtocounter{equation}{1}  
    and the weight distribution of $\C_k(\F_q,{\bf v},\infty)^{\perp}$ is 
    given by Equation (\ref{eq.weight.6})
    % The previous equation was number 12.
    % Account for the double column equations here.
    \addtocounter{equation}{1}  
\end{corollary}

\begin{figure*}[!t]
	% ensure that we have normalsize text
	\normalsize
	% Store the current equation number.
	\setcounter{MYtempeqncnt}{\value{equation}}
	% Set the equation number to one less than the one
	% desired for the first equation here.
	% The value here will have to changed if equations
	% are added or removed prior to the place these
	% equations are referenced in the main text.
	\setcounter{equation}{11}
    \begin{align}\label{eq.weight.5}
        \small
        A_{i}=\left\{\begin{array}{ll}
            1, & {\rm if}~i=0, \\
            0, & {\rm if}~i\in \{1,2,\ldots,q-k\}, \\ 
            \frac{q-1}{q}\binom{q}{k}, & {\rm if}~i=q-k+1~{\rm and}~p\nmid k, \\ 
            \frac{q-1}{q}\left[\binom{q}{k}+(-1)^{k+\frac{k}{p}} (q-1)\binom{\frac{q}{p}}{\frac{k}{p}} \right], & {\rm if}~i=q-k+1~{\rm and}~p\mid k,  \\
            \binom{q+1}{k-s}\sum_{j=0}^{s-1}(-1)^j \binom{q-k+1+s}{j}(q^{s-j}-1)+(-1)^s \binom{k}{s}A_{q-k+1}, & {\rm if}~i=q-k+s+1,~s\in \{1,2,\ldots,k\}. 
        \end{array} \right.
    \end{align}
	% Restore the current equation number.
	\setcounter{equation}{\value{MYtempeqncnt}}
	% The IEEE uses as a separator
	\hrulefill
	% The spacer can be tweaked to stop underfull vboxes.
	\vspace*{4pt}
\end{figure*}

\begin{figure*}[!t]
	% ensure that we have normalsize text
	\normalsize
	% Store the current equation number.
	\setcounter{MYtempeqncnt}{\value{equation}}
	% Set the equation number to one less than the one
	% desired for the first equation here.
	% The value here will have to changed if equations
	% are added or removed prior to the place these
	% equations are referenced in the main text.
	\setcounter{equation}{12}
    \begin{align}\label{eq.weight.6}
        \small
        A_{i}^{\perp}=\left\{\begin{array}{ll}
            1, & {\rm if}~i=0, \\
            0, & {\rm if}~i\in \{1,2,\ldots,k-1\}, \\ 
            \frac{q-1}{q}\binom{q}{k}, & {\rm if}~i=k~{\rm and}~p\nmid k, \\ 
            \frac{q-1}{q}\left[\binom{q}{k}+(-1)^{k+\frac{k}{p}} (q-1)\binom{\frac{q}{p}}{\frac{k}{p}} \right], & {\rm if}~i=k~{\rm and}~p\mid k, \\
            \binom{q+1}{k+s}\sum_{j=0}^{s-1}(-1)^j \binom{k+s}{j}(q^{s-j}-1)+(-1)^s\binom{k}{s}A_{k}^{\perp}, & {\rm if}~i=k+s,~s\in \{1,2,\ldots,q-k+1\}. 
        \end{array} \right.
    \end{align}
	% Restore the current equation number.
	\setcounter{equation}{\value{MYtempeqncnt}}
	% The IEEE uses as a separator
	\hrulefill
	% The spacer can be tweaked to stop underfull vboxes.
	\vspace*{4pt}
\end{figure*}

\begin{proof}
    Combining with \cite[Corollary 1]{HF2023}, by a similar argument to that of Corollary \ref{coro.NMDS Fq*}, the desired result holds.  
\end{proof}

In Remark \ref{rem.111}, we have pointed out that the class of linear codes $\C_k(\SSS,{\bf v},\infty)$ are not equivalent to 
the (+)-extended TGRS codes $\C_{1,k-1,k,n}(\SSS, {\bf v}, \eta, \infty)$ studied in \cite{ZL2024} generally. 
In the following remark, we present some detailed comparison results based on the weight distributions to support this statement.   
%We emphasize that these facts confirm again that (+)-extended TGRS codes are generally not equivalent to $\C_k(\SSS,{\bf v},\infty)$. 

\begin{remark}\label{rem.weight distribution} Some comparison remarks between results in \cite{ZL2024} and ours are given as follows. 
    \begin{enumerate}
        \item [\rm 1)] Compared to \cite[Theorem 3.3]{ZL2024}, the weight distributions of the NMDS (+)-extended TGRS codes depend on 
        the values of $N(k,-\eta^{-1},\SSS)$, while our weight distributions presented in Theorem \ref{th.NMDS weight distribution} are related to $N(k,0,\SSS)$. 
        Since $\eta\neq 0$, then $N(k,-\eta^{-1},\SSS)$ may not equal $N(k,0,\SSS)$. 
        Hence, $\C_k(\SSS,{\bf v},\infty)$ and $\C_{1,k-1,k,n}(\SSS, {\bf v}, \eta, \infty)$ may have different weight distributions. 
        Moreover,  the difference of $A_{n-k+1}=A_k^{\perp}$ in $\C_k(\SSS,{\bf v},\infty)$ and  $\C_{1,k-1,k,n}(\SSS, {\bf v}, \eta, \infty)$
        is  $$(q-1)\left(N(k,0,\SSS)-N(k,-\eta^{-1},\SSS)\right).$$
        %Also this fact implies that (+)-extended TGRS codes are generally not equivalent to $\C_k(\SSS,{\bf v},\infty)$. 

        \item [\rm 2)] Compared to \cite[Corollary 3.2]{ZL2024} ($\SSS=\F_q$), the $[q+1,k,q-k+1]_q$ NMDS codes derived from 
        the (+)-extended TGRS codes $\C_{1,k-1,k,n}({\F_q}, {\bf v}, \eta, \infty)$ 
        and the class of linear codes $\C_k(\F_q,{\bf v},\infty)$ 
        have the same weight distribution if $p\nmid k$ and have different weight distributions if $p\mid k$. 
        In fact, it can be verified that the value of $A_{q-k+1}=A_k^{\perp}$ between them differs by 
        $$(-1)^{k+\frac{k}{p}}(q-1)\binom{\frac{q}{p}}{\frac{k}{p}}~{\rm for}~ p\mid k.$$

        \item [\rm 3)] Compared to \cite[Corollary 3.3]{ZL2024} ($\SSS=\F_q^*$), the $[q,k,q-k]_q$ NMDS codes derived from 
        the (+)-extended TGRS codes $\C_{1,k-1,k,n}({\F_q^*}, {\bf v}, \eta, \infty)$ 
        and the class of linear codes $\C_k(\F_q^*,{\bf v},\infty)$ 
        have different weight distributions.  
        It is easily seen that the value of $A_{q-k}=A_k^{\perp}$ actually differs by 
        $$(-1)^{k+\lfloor \frac{k}{p} \rfloor}(q-1)\binom{\frac{q}{p}-1}{\lfloor \frac{k}{p} \rfloor}.$$
    \end{enumerate}
\end{remark}

In the following, we give three concrete examples to visualize the results obtained so far. 
The first one is a non-GRS  MDS code. 
The second and third ones are NMDS codes that are not equivalent to (+)-extended TGRS codes constructed in \cite{ZL2024}. 
%All these results have also been double-checked by Magma \cite{Magma}. 
% based on \cite[Lemma 10]{BPR2022}.   
%(a complex method based on determinants of $2$-order and $3$-order submatrices of a certain matrix associated with the standard generator matrix of a given linear code). 

\begin{example}\label{exam.1 MDS and NMDS codes from EGRS}
    Let $\omega$ be a primitive element of $\F_{16}$ satisfying $\w^4+\w+1=0$. 
    Let $\SSS=\{\omega, \omega^3, \omega^5, \omega^8, \omega^9, \omega^{11}, \omega^{13}\}\subseteq \F_{16}$. 
    Then it is easy to verify that $\SSS$ is $5$-zero-sum-free. 
    %$\SSS_1$ is $4$-zero-sum free and $\SSS_2$ contains a $4$-zero-sum subset such as $\{1, \omega, 2, \omega^5, \omega^3\}$. 
    Therefore, for any ${\bf v}\in (\F_{16}^*)^7$, it follows from Theorems \ref{th.length and dimension}, \ref{th.MDS condition} 1) and \ref{th.non-GRS} 
    that $\C_5(\SSS,{\bf v},\infty)$ is an $[8, 5, 4]_{16}$ non-GRS  MDS code. 
    Double-checked by Magma \cite{Magma}, all the statements are true. 
    %By Magma \cite{Magma}, the standard generator matrix of $\C_k(\SSS,{\bf v},\infty)$ is  
    %$$(I_5~~A)~{\rm with}~A=\left(\begin{array}{ccc}
    %    \omega^8 & \omega^8 & \omega^2 \\
    %    \omega^{13} & \omega^{10} & \omega^{13} \\
    %    \omega^2  &  1 & \omega^{10} \\
    %    \omega^2 & \omega^3 & \omega^4 \\
    %    \omega^{14} & \omega^9 & \omega^{13} 
    %\end{array}\right).$$ 
    %Double-checked by \cite[Lemma 10]{BPR2022}, $\C_5(\SSS,{\bf v},\infty)$ is indeed non-GRS since 
    % $$\det\left(\left(\begin{array}{ccc}
    %    (\omega^8)^{-1} & (\omega^8)^{-1} & (\omega^2)^{-1} \\
    %    (\omega^{13})^{-1} & (\omega^{10})^{-1} & (\omega^{13})^{-1} \\
    %    (\omega^2)^{-1}  &  1^{-1} & (\omega^{10})^{-1} 
    %\end{array}\right)\right)=\det\left(\left(\begin{array}{ccc}
    %    \omega^7 & \omega^7 & \omega^{13} \\
    %    \omega^{2} & \omega^{5} & \omega^{2} \\
    %    \omega^{13}  &  1 & \omega^{5}
    %\end{array}\right)\right)=\omega^{12}\neq 0.$$ 
\end{example}

\begin{example}\label{exam.2 MDS and NMDS codes from EGRS}
    Let $\omega$ be a primitive element of $\F_9$ satisfying $\w^2+2\w+2=0$. 
    %Let $\SSS_1=\F_9^*$ and $\SSS_2=\F_9$.  
    We have the following results.  
    \begin{itemize}
        \item [\rm 1)] For any vector ${\bf v}_1\in (\F_9^*)^8$, it follows from Corollary \ref{coro.NMDS Fq*} that 
        $\C_4(\F_9^*,{\bf v}_1, \infty)$ is a $[9, 4, 5]_9$ NMDS code 
        with the weight enumerator 
        \begin{small}
            \begin{align*}
                1+48z^5+480z^6+1152z^7+2616z^8+2264z^9.
            \end{align*}
        \end{small}
        %$1+48z^5+480z^6+1152z^7+2616z^8+2264z^9.$ 
        %{\color{red}$\dim(\Hull(C))=3$.}

        \item [\rm 2)] For any vector ${\bf v}_2\in (\F_9^*)^9$, it follows from Corollary \ref{coro.NMDS Fq} that 
        $\C_6(\F_9,{\bf v}_2, \infty)$ is a $[10, 6, 4]_9$ NMDS code 
        with the weight enumerator 
        \begin{small}
            \begin{align*}
                & 1+96z^4+1440z^5+8160z^6+38400z^7+\\
                & 115200z^8+204464z^9+163680z^{10}.
            \end{align*}
        \end{small}
        %$1+96z^4+1440z^5+8160z^6+38400z^7+115200z^8+204464z^9+163680z^{10}.$ 
        %{\color{red}$\dim(\Hull(C))=4$.}
    \end{itemize}
    In addition, both $\C_4(\F_9^*,{\bf v}_1, \infty)$ and $\C_6(\F_9,{\bf v}_2, \infty)$ are not equivalent to the (+)-extended TGRS codes 
    constructed in \cite{ZL2024} according to Remarks \ref{rem.weight distribution} 2) and 3). 
    Verified by Magma \cite{Magma}, all these results are true. 
\end{example}

%\subsection{Several new infinite families of NMDS codes with dimensions $5$ and $6$} 

Note also that Heng $et~al.$ \cite{HW2023} constructed three infinite families of $[q,k,q-k]_q$ NMDS codes for $k\in \{5,6\}$ and $q=2^m$, 
where $m\geq 5$ is odd (two classes) or $m\geq 4$ is an integer (one class). 
By taking $k=5$ and $6$ in Theorem \ref{th.must be MDS or NMDS}, respectively, we can immediately obtain infinite families of $[q,k,q-k]_q$ NMDS codes for any $q> 8$. 
Moreover, their weight distributions can be explicitly determined by Corollary \ref{coro.NMDS Fq*}. 
To emphasize the validity and importance of our results, we collect these NMDS codes in the following theorems.

\begin{theorem}\label{th.5 dimensional NMDS codes}
    Let $q=p^m>8$ and  $\C_5(\F_q^*,{\bf v},\infty)$ be the $q$-ary linear code generated by $G_5$ given in Equation (\ref{eq.generator matrix}). 
    Then $\C_5(\F_q^*,{\bf v},\infty)$ is a $[q,5,q-5]_q$ NMDS code with the following weight distribution.  
    %Moreover, the following statements hold.   
    \begin{enumerate}
        \item [\rm 1)] If $p=2$, then the weight distribution of $\C_5(\F_q^*,{\bf v},\infty)$ is given in Table \ref{tab:1}. 
        \item [\rm 2)] If $p\in \{3, 5\}$, then the weight distribution of $\C_5(\F_q^*,{\bf v},\infty)$ is given in Table \ref{tab:2}. 
        \item [\rm 3)] If $p\geq 7$, then the weight distribution of $\C_5(\F_q^*,{\bf v},\infty)$ is given in Table \ref{tab:3}. 
        
        \begin{table}
            \centering
            \caption{The weight distribution of the $[q,5,q-5]_q$ NMDS code $\C_5(\F_q^*,{\bf v},\infty)$ for $p=2$}\label{tab:1}       % Give a unique label
            \vspace{-2mm}
            \begin{tabular}{ll}
             \hline              
            Weight $i$ & Multiplicity $A_i$ \\ \hline \hline 
            $0$ &  $1$ \\ 
            $q-5$ & $(q-1)^2(q-2)(q-4)(q-8)/120$ \\ 
            $q-4$ & $(q-1)^2(q-2)(9q-32)/24$ \\ 
            $q-3$ & $(q-1)^2(q-2)(q^2-4q+32)/12$ \\ 
            $q-2$ & $(q-1)^2(2q^3+11q^2-20q+64)/12$ \\ 
            $q-1$ & $(q-1)(9q^4+9q^3+38q^2-24q+64)/24$ \\ 
            $q$ & $(q-1)(44q^4+25q^3+5q^2-10q+56)/120$ \\ \hline
        \end{tabular}

        %\begin{align*}
        %    \small
        %    A(z) = & 1 + \frac{(q-1)^2(q-2)(q-4)(q-8)}{120}z^{q-5} + \frac{(q-1)^2(q-2)(9q-32)}{24}z^{q-4} + \\ 
        %           & \frac{(q-1)^2(q^3-6q^2+40q-64)}{12}z^{q-3} + \frac{(q-1)^2(2q^3+11q^2-20q+64)}{12}z^{q-2} + \\
        %           & \frac{(q-1)(9q^4+9q^3+38q^2-24q+64)}{24}z^{q-1} + \frac{(q-1)(44q^4+25q^3+5q^2-10q+56)}{120}z^q. 
        %\end{align*}

        \vspace{1mm}

            \centering
            \caption{The weight distribution of the $[q,5,q-5]_q$ NMDS code $\C_5(\F_q^*,{\bf v},\infty)$ for $p\in \{3, 5\}$}\label{tab:2}
            \vspace{-2mm}
            \scalebox{0.7}{\begin{tabular}{ll}
             \hline              
            Weight $i$ & Multiplicity $A_i$ \\ \hline \hline 
            $0$ &  $1$ \\ 
            $q-5$ & $(q-1)^2(pq^3-14pq^2+71pq-154p+120)/(120p)$ \\ 
            $q-4$ & $(q-1)^2(9pq^2-65pq+154p-120)/(24p)$ \\ 
            $q-3$ & $(q-1)^2(pq^3-6pq^2+55pq-154p+120)/(12p)$ \\ 
            $q-2$ & $(q-1)^2(2pq^3+11pq^2-35pq+154p-120)/(12p)$ \\ 
            $q-1$ & $(q-1)(9pq^4+9pq^3+53pq^2-129pq+120q+154p-120)/(24p)$ \\ 
            $q$ & $(q-1)(44pq^4+25pq^3-10pq^2+95pq-120q-34p+120)/(120p)$ \\ \hline
        \end{tabular}}
        
        \vspace{1mm}

        %\begin{align*}
        %    \small
        %    A(z) = & 1 + \frac{(q-1)^2(pq^3-14pq^2+71pq-154p+120)}{120p}z^{q-5} + \\
        %           & \frac{(q-1)^2(9pq^2-65pq+154p-120)}{24p}z^{q-4} + \\ 
        %           & \frac{(q-1)^2(pq^3-6pq^2+55pq-154p+120)}{12p}z^{q-3} + \\
        %           & \frac{(q-1)^2(2pq^3+11pq^2-35pq+154p-120)}{12p}z^{q-2} + \\
        %           & \frac{(q-1)(9pq^4+9pq^3+53pq^2-129pq+120q+154p-120)}{24p}z^{q-1} + \\ 
        %           & \frac{(q-1)(44pq^4+25pq^3-10pq^2+95pq-120q-34p+120)}{120p}z^q. 
        %\end{align*}
            \centering
            \caption{The weight distribution of the $[q,5,q-5]_q$ NMDS code $\C_5(\F_q^*,{\bf v},\infty)$ for $p\geq 7$}\label{tab:3}
            \vspace{-2mm}
            \begin{tabular}{ll}
             \hline              
            Weight $i$ & Multiplicity $A_i$ \\ \hline \hline 
            $0$ &  $1$ \\ 
            $q-5$ & $(q-1)^2(q-7)(q^2-7q+22)/120$ \\ 
            $q-4$ & $(q-1)^2(9q^2-65q+154)/24$ \\ 
            $q-3$ & $(q-1)^2(q^3-6q^2+55q-154)/12$ \\ 
            $q-2$ & $(q-1)^2(2q^3+11q^2-35q+154)/12$ \\ 
            $q-1$ & $(q-1)(9q^4+9q^3+53q^2-129q+154)/24$ \\ 
            $q$ & $(q-1)(44q^4+25q^3-10q^2+95q-34)/120$ \\ \hline
        \end{tabular}
        \end{table}
    \end{enumerate}
\end{theorem}
\begin{proof}
    Note that $\lfloor \frac{5}{p} \rfloor=2$ if $p=2$; $\lfloor \frac{5}{p} \rfloor=1$ if $p\in \{3, 5\}$; and $\lfloor \frac{5}{p} \rfloor=0$ if $p\geq 7$. 
    Then the desired results follow from Corollary \ref{coro.NMDS Fq*} after a direct but tedious calculation.  
\end{proof}

\begin{theorem}\label{th.6 dimensional NMDS codes}
    Let $q=p^m>8$ and $\C_6(\F_q^*,{\bf v},\infty)$ be the $q$-ary linear code generated by $G_6$ given in Equation (\ref{eq.generator matrix}). 
    Then $\C_6(\F_q^*,{\bf v},\infty)$ is a $[q,6,q-6]_q$ NMDS code with the following weight distribution. 
    \begin{enumerate}
        \item [\rm 1)] If $p=2$, then the weight distribution of $\C_6(\F_q^*,{\bf v},\infty)$ is given in Table \ref{tab:4}. 
        \item [\rm 2)] If $p=3$, then the weight distribution of $\C_6(\F_q^*,{\bf v},\infty)$ is given in Table \ref{tab:5}. 
        \item [\rm 3)] If $p=5$, then the weight distribution of $\C_6(\F_q^*,{\bf v},\infty)$ is given in Table \ref{tab:6}. 
        \item [\rm 4)] If $p\geq 7$, then the weight distribution of $\C_6(\F_q^*,{\bf v},\infty)$ is given in Table \ref{tab:7}. 
    \end{enumerate}
\end{theorem}
\begin{proof}
    Note that $\lfloor \frac{6}{p} \rfloor=3$ if $p=2$; $\lfloor \frac{6}{p} \rfloor=2$ if $p=3$; 
    $\lfloor \frac{6}{p} \rfloor=1$ if $p=5$;  and $\lfloor \frac{6}{p} \rfloor=0$ if $p\geq 7$.  
    Then the desired results again follow from Corollary \ref{coro.NMDS Fq*} after a direct but tedious calculation.  
\end{proof}

\begin{table}
    \centering
    \caption{The weight distribution of the $[q,6,q-6]_q$ NMDS code $\C_6(\F_q^*,{\bf v},\infty)$ for $p=2$}\label{tab:4}       % Give a unique label
    \vspace{-3mm}
    \scalebox{0.85}{\begin{tabular}{ll}
     \hline              
    Weight $i$ & Multiplicity $A_i$ \\ \hline \hline 
    $0$ &  $1$ \\ 
    $q-6$ & $(q-1)^2(q-2)(q-4)(q-6)(q-8)/720$ \\ 
    $q-5$ & $(q-1)^2(q-2)(q-4)(11q-48)/120$ \\ 
    $q-4$ & $(q-1)^2(q-2)(q^3-8q^2+74q-192)/48$ \\ 
    $q-3$ & $(q-1)^2(q-2)(2q^3+15q^2-44q+192)/36$ \\ 
    $q-2$ & $(q-1)^2(9q^4+16q^3+96q^2-160q+384)/48$ \\ 
    $q-1$ & $(q-1)(44q^5+71q^4+35q^3+250q^2-184q+384)/120$ \\ 
    $q$ & $(q-1)(265q^5+129q^4+100q^3+30q^2-140q+336)/720$ \\ \hline
\end{tabular}}

    \centering
    \caption{The weight distribution of the $[q,6,q-6]_q$ NMDS code $\C_6(\F_q^*,{\bf v},\infty)$ for $p=3$}\label{tab:5}
    \vspace{-3mm}
    \scalebox{0.85}{\begin{tabular}{ll}
     \hline              
    Weight $i$ & Multiplicity $A_i$ \\ \hline \hline 
    $0$ &  $1$ \\ 
    $q-6$ & $(q-1)^2(q-3)(q-6)(q^2-11q+38)/720$ \\ 
    $q-5$ & $(q-1)^2(q-3)(11q^2-96q+228)/120$ \\ 
    $q-4$ & $(q-1)^2(q-3)(q^3-7q^2+84q-228)/48$ \\ 
    $q-3$ & $(q-1)^2(2q^4+11q^3-89q^2+420q-684)/36$ \\ 
    $q-2$ & $(q-1)^2(9q^4+16q^3+111q^2-300q+684)/48$ \\ 
    $q-1$ & $(q-1)(44q^5+71q^4+20q^3+405q^2-624q+684)/120$ \\ 
    $q$ & $(q-1)(265q^5+129q^4+115q^3-125q^2+300q+36)/720$ \\ \hline
\end{tabular}}

\centering
\caption{The weight distribution of the $[q,6,q-6]_q$ NMDS code $\C_6(\F_q^*,{\bf v},\infty)$ for $p=5$}\label{tab:6}
\vspace{-3mm}
\scalebox{0.85}{\begin{tabular}{ll}
 \hline              
Weight $i$ & Multiplicity $A_i$ \\ \hline \hline 
$0$ &  $1$ \\ 
$q-6$ & $(q-1)^2(q-5)(q^3-15q^2+80q-180)/720$ \\ 
$q-5$ & $(q-1)^2(11q^3-129q^2+556q-900)/120$ \\ 
$q-4$ & $(q-1)^2(q^4-10q^3+105q^2-520q+900)/48$ \\ 
$q-3$ & $(q-1)^2(2q^4+11q^3-89q^2+460q-900)/36$ \\ 
$q-2$ & $(q-1)^2(9q^4+16q^3+111q^2-340q+900)/48$ \\ 
$q-1$ & $(q-1)(44q^5+71q^4+20q^3+445q^2-880q+900)/120$ \\ 
$q$ & $(q-1)(265q^5+129q^4+115q^3-165q^2+556q-180)/720$ \\ \hline
\end{tabular}}

\centering
    \caption{The weight distribution of the $[q,6,q-6]_q$ NMDS code $\C_6(\F_q^*,{\bf v},\infty)$ for $p\geq 7$}\label{tab:7}
    \vspace{-3mm}
    \scalebox{0.85}{\begin{tabular}{ll}
     \hline              
    Weight $i$ & Multiplicity $A_i$ \\ \hline \hline 
    $0$ &  $1$ \\ 

    $q-6$ & $(q-1)^2(q^4-20q^3+155q^2-580q+1044)/720$ \\ 

    $q-5$ & $(q-1)^2(11q^3-129q^2+556q-1044)/120$ \\ 

    $q-4$ & $(q-1)^2(q^4-10q^3+105q^2-520q+1044)/48$ \\ 

    $q-3$ & $(q-1)^2(2q^4+11q^3-89q^2+460q-1044)/36$ \\ 

    $q-2$ & $(q-1)^2(9q^4+16q^3+111q^2-340q+1044)/48$ \\ 

    $q-1$ & $(q-1)(44q^5+71q^4+20q^3+445q^2-1024q+1044)/120$ \\ 

    $q$ & $(q-1)(265q^5+129q^4+115q^3-165q^2+700q-324)/720$ \\ \hline
\end{tabular}}
\end{table}

\begin{remark}\label{rem.weight.Heng} $\quad$ 
    \begin{enumerate}
        \item [\rm 1)] Note that our infinite families of $[q,5,q-5]_q$ NMDS codes and $[q,6,q-6]_q$ NMDS codes always exist for any $q=p^m>8$, 
        while the infinite family of $[q,5,q-5]_q$ NMDS codes constructed in \cite[Theorem 31]{HW2023} and two infinite families of 
        $[q,6,q-6]_q$ NMDS codes constructed in \cite[Theorems 38 and 42]{HW2023} only exist for $q=2^m$. 
        Hence, our NMDS codes presented in Theorems \ref{th.5 dimensional NMDS codes} and \ref{th.6 dimensional NMDS codes} have more flexible parameters. 

        \item [\rm 2)] Comparing Table \ref{tab:1} with \cite[Theorem 31]{HW2023} and Table \ref{tab:4} with \cite[Theorems 38 and 42]{HW2023}, 
        we can also conclude that our NMDS codes are not equivalent to those constructed by Heng $et~al.$ \cite{HW2023} even for 
        $q=2^m$ because of different weight distributions.   
    \end{enumerate}
\end{remark}

\subsection{Self-orthogonal properties of $\C_k(\SSS,{\bf v},\infty)$}  

Now we characterize the self-orthogonal properties  of $\C_k(\SSS,{\bf v},\infty)$. 

%In other words, we will present some conditions for $\C_k(\SSS,{\bf v},\infty)$ being self-orthogonal or having $(k-1)$-dimensional hull. 

\begin{theorem}\label{th.so criterion}
    Let $\C_k(\SSS,{\bf v},\infty)$ be the $q$-ary linear code generated by $G_k$ given in Equation (\ref{eq.generator matrix}). 
    Then $\C_k(\SSS,{\bf v},\infty)$ is self-orthogonal if and only if there exists a polynomial $g(x)=\sum_{s=0}^{n-2k}g_sx^s\in \F_q[x]_{n-2k+1}$ 
    such that the following two conditions hold: 
    \begin{itemize}
        \item [\rm 1)] $v_i^2=u_ig(a_i)$ for any $1\leq i\leq n$; 
        \item [\rm 2)] $g_{n-2k-1}+g_{n-2k}\sum_{i=1}^{n}a_i=-1$.
    \end{itemize}
    %\begin{align}
    %   v_i^2=u_ih(a_i)~{\rm and}~h_{n-2k-1}+h_{n-2k}\sum_{i=1}^{n}a_i=-1 
    %\end{align}
    %for any $1\leq i\leq n$. 
\end{theorem}
\begin{proof}
    According to Lemma \ref{lem.so condition}, $\C_k(\SSS,{\bf v},\infty)$ is self-orthogonal if and only if $G_kG_k^T$ is a $k\times k$ zero matrix, 
    which is equivalent to 
    \begin{align}\label{eq.SO condition}
        \small
        \left\{\begin{array}{ll}
            \sum_{i=1}^{n} v_i^2 a_i^j = 0, ~0\leq j\leq 2k-2, &{\rm (\ref*{eq.SO condition}.1)} \\ \\ 
            \sum_{i=1}^{n} v_i^2 a_i^{2k} + 1 = 0. & {\rm (\ref*{eq.SO condition}.2)}
        \end{array} \right.
    \end{align}
    
    %``$(\Rightarrow)$'' %Denote by $x_i=v_i^2$ for each $1\leq i\leq n$.  
    Assume that $\C_k(\SSS,{\bf v},\infty)$ is self-orthogonal and consider the system of equation 
    \begin{align}\label{eq.system equation}
        \sum_{i=1}^{n}a_i^jx_i=0~{\rm for}~0\leq j\leq 2k-2. 
    \end{align}
    %$\sum_{i=1}^{n}a^jx_i=0$, where $0\leq j\leq 2k-2$. 
    With Lemma \ref{lem.PRA}, we conclude that %$(u_1,u_2,\ldots,u_n)$, $(a_1u_1,a_2u_2,\ldots,a_nu_n)$, $(a_1^2u_1,a_2^2u_2,\ldots,a_n^2u_n)$, $\ldots$, $(a_1^{n-2k}u_1,a_2^{n-2k}u_2,\ldots,a_n^{n-2k}u_n)$ 
    $
     (u_1,\ldots,u_n),\\(a_1u_1,\ldots,a_nu_n), \ldots,(a_1^{n-2k}u_1,\ldots,a_n^{n-2k}u_n)
    $
    are $n-2k+1$ linearly independent solutions of the system of equation (\ref{eq.system equation}). 
    On the other hand, it is easy to see that the rank of the coefficient matrix of the system of equation (\ref{eq.system equation}) is $2k-1$. 
    Since $(n-2k+1)+(2k-1)=n$, then all these $n-2k+1$ solutions form a basis of the solution space of the system of equation (\ref{eq.system equation}). 
    It then turns out from Equation (\ref{eq.SO condition}.1) that there are $g_0, \ldots, g_{n-2k}\in \F_q$ such that 
    \begin{align}\label{eq.vi}
        v_i^2=\sum_{s=0}^{n-2k}g_sa_i^su_i~{\rm for~each}~1\leq i\leq n. 
    \end{align}
    Denote by $g(x)=\sum_{s=0}^{n-2k}g_sx^s$. Then Equation (\ref{eq.vi}) implies that $v_i^2=u_ig(a_i)$ for each $1\leq i\leq n$. 
    Moreover, it follows from substituting Equation (\ref{eq.vi})  to Equation (\ref{eq.SO condition}.2) and Lemma \ref{lem.PRA} that 
    \begin{align*}
        &\sum_{i=1}^n \left(\sum_{s=0}^{n-2k}g_sa_i^su_i\right) a_i^{2k}+1 \\
    =&\sum_{s=0}^{n-2k} g_s \left( \sum_{i=1}^{n}a_i^{s+2k}u_i\right)+1 \\
    =&g_{n-2k-1}+g_{n-2k}\sum_{i=1}^{n}a_i+1 \\
    =&  0,
    \end{align*}
    %$\sum_{i=1}^n \left(\sum_{s=0}^{n-2k}g_sa_i^su_i\right) a_i^{2k}+1
    %=\sum_{s=0}^{n-2k} g_s \left( \sum_{i=1}^{n}a_i^{s+2k}u_i\right)+1
    %=g_{n-2k-1}+g_{n-2k}\sum_{i=1}^{n}a_i+1
    %=  0,$ 
    %$$\sum_{i=1}^n (\sum_{s=0}^{n-2k}h_sa_i^su_i) a_i^{2k}+1 = \sum_{s=0}^{n-2k} g_s \sum_{i=1}^{n}a_i^{s+2k}u_i+1=g_{n-2k-1}+g_{n-2k}\sum_{i=1}^{n}a_i+1=0,$$ 
    which deduces that $g_{n-2k-1}+g_{n-2k}\sum_{i=1}^{n}a_i=-1$. This completes the proof of sufficiency. 

    Conversely, if there exists a polynomial $g(x)=\sum_{s=0}^{n-2k}g_sx^s\in \F_q[x]_{n-2k+1}$ such that the conditions 1) and 2) hold, 
    then it is easy to verify that Equation (\ref{eq.SO condition}) holds. 
    Hence, $\C_k(\SSS,{\bf v},\infty)$ is self-orthogonal. This completes the proof of necessity. 
    
    Combining these two aspects, we complete the proof.  
\end{proof}

%Following Theorem \ref{th.so criterion}, we further study the (almost) self-dual properties of $\C_k(\SSS,{\bf v},\infty)$. 
%By Theorem \ref{th.so criterion}, we conclude that there are no self-dual codes in $\C_k(\SSS,{\bf v},\infty)$ as follows. 

\begin{corollary}\label{coro.no SD codes}
    Let $n$ be odd and $\C_\frac{n+1}{2}(\SSS,{\bf v},\infty)$ be the $[n+1,\frac{n+1}{2}]_q$ linear code generated by $G_\frac{n+1}{2}$ given in Equation (\ref{eq.generator matrix}). 
    Then $\C_\frac{n+1}{2}(\SSS,{\bf v},\infty)$ can not be self-dual for any $\SSS\subseteq \F_q$ and ${\bf v}\in (\F_q^*)^n$. 
\end{corollary}
\begin{proof}
    The desired result follows straightforward from taking $k=\frac{n+1}{2}$ in Theorem \ref{th.so criterion} and the fact that $n-2k=-1<0$.  
\end{proof}

%\begin{remark}
%    Although Han and Zhang constructed many classes of self-dual codes in \cite{HZ2023}, 
%    Corollary \ref{coro.no SD codes} implies that there are no self-dual codes in $\C_k(\SSS,{\bf v},\infty)$. 
%    Recall that $\C_k(\SSS,{\bf v},\infty)$ is obtained by adding a column vector of the form $(0,0,\ldots,0,1)$ to 
%    a generator matrix of the linear code constructed in \cite{HZ2023}. 
%    Therefore, these facts lead to a clear distinction between $\C_k(\SSS,{\bf v},\infty)$ and GRS codes.
%\end{remark}

\begin{corollary}\label{coro.ASD codes}
    Let $n$ be even and $\C_\frac{n}{2}(\SSS,{\bf v},\infty)$ be the $[n+1,\frac{n}{2}]_q$ linear code generated by $G_\frac{n}{2}$ given in Equation (\ref{eq.generator matrix}). 
    Then $\C_\frac{n}{2}(\SSS,{\bf v},\infty)$ is almost self-dual if and only if there is a $\lambda\in \F_q^*$ such that 
    $v_i^2=\lambda u_i$ and $\sum_{i=1}^{n}a_i=-\lambda^{-1}$. 
\end{corollary}
\begin{proof}
    The desired result follows straightforward from taking $k=\frac{n}{2}$ in Theorem \ref{th.so criterion}.  
\end{proof}

According to Corollary \ref{coro.ASD codes}, we give two explicit constructions of almost self-dual codes. 

\begin{theorem}\label{th.ASD}
    Let $q=2^m\geq 8$ and $3\leq k\leq  \frac{q}{2}$. 
    Suppose that $\SSS=\{a_1,\ldots,a_{2k}\}$ is a $2k$-subset of $\F_q$. 
    If $\sum_{i=1}^{2k}a_i\neq 0$, we let $v_i=(\lambda u_i)^{\frac{q}{2}}$ 
    for any $1\leq i\leq 2k$ and $\lambda =(\sum_{i=1}^{2k}a_i)^{-1}$, 
    then $\C_{k}(\SSS, {\bf v},\infty)$ is a $[2k+1,k]_q$ non-GRS  MDS or NMDS almost self-dual code. 
\end{theorem}
\begin{proof}
    Since $q$ is even and $\sum_{i=1}^{2k}a_i\neq 0$, then $\lambda =(\sum_{i=1}^{2k}a_i)^{-1}\neq 0$ 
    and $v_i=(\lambda u_i)^{\frac{q}{2}}\neq 0$ can be taken for any $1\leq i\leq 2k$. 
    Moreover, it follows from $\lambda u_i\in \F_q$ that $v_i^2=((\lambda u_i)^{\frac{q}{2}})^2=\lambda u_i$ for each $1\leq i\leq 2k$.  
    By Theorems \ref{th.must be MDS or NMDS} and \ref{th.non-GRS} as well as Corollary \ref{coro.ASD codes}, 
    we immediately get the desired result.   
\end{proof}

\begin{theorem}\label{th.ASD2}
    Let $r$ and $m$ be two positive integers satisfying $r\mid m$ and $2\mid \frac{m}{r}$. 
    Let $q=p^m$ and $3\leq k\leq  \frac{p^r}{2}$. 
    Suppose that $\SSS=\{a_1,\ldots,a_{2k}\}$ is a $2k$-subset of $\F_{p^r}$. 
    If $\sum_{i=1}^{2k}a_i\neq 0$, we let $v_i=(\lambda u_i)^{\frac{1}{2}}$ 
    for any $1\leq i\leq 2k$ and $\lambda =(-\sum_{i=1}^{2k}a_i)^{-1}$, 
    then   $\C_{k}(\SSS, {\bf v},\infty)$ is a $[2k+1,k]_q$ non-GRS  MDS or NMDS almost self-dual code. 
\end{theorem}
\begin{proof}
    Note that $\sum_{i=1}^{2k}a_i\in \F_{p^r}^*$ and $u_i\in \F_{p^r}^*$ for any $1\leq i\leq 2k$.  
    Note also that each element in $\F_{p^r}$ is a square element in $\F_{p^m}$ since $r\mid m$ and $2\mid \frac{m}{r}$.  
    Then $\lambda =(-\sum_{i=1}^{2k}a_i)^{-1}\neq 0$ and $v_i=(\lambda u_i)^{\frac{1}{2}}\in \F_{p^m}^*$, 
    $i.e.$, $v_i^2= \lambda u_i$ can be taken for each $1\leq i\leq 2k$.  
    Then the desired result again follows straightforward from Theorems \ref{th.must be MDS or NMDS} and \ref{th.non-GRS} 
    as well as Corollary \ref{coro.ASD codes}.
\end{proof}

By now, we have answered {\bf Problem \ref{prob2}}.
In the following, we give two explicit examples of almost self-dual codes from {Theorems \ref{th.ASD} and \ref{th.ASD2}}.

\begin{example}\label{exam.asd1}
    Let $\omega$ be a primitive element of $\F_8$ satisfying $\w^3+\w+1=0$ and 
    $\SSS=\{1,\omega,\omega^2,\omega^3,\omega^4,\omega^6\}\subseteq \F_8.$  
    Note that $\sum_{a\in\SSS}a=\omega^5$ and $\SSS$ contains a $3$-zero-sum subset such as  $\{1, \omega, \omega^3 \}$.  
    Take $v_i=(\lambda u_i)^4$ for each $1\leq i\leq 6$, where $\lambda=(\sum_{a\in\SSS}a)^{-1}=\omega^2$, 
    and hence, we get that ${\bf v}=(\omega^3,\omega,1,1,\omega^3,\omega)\subseteq (\F_8^*)^6.$
    By Theorems \ref{th.MDS condition} 2) and \ref{th.ASD}, $\C_{k}(\SSS, {\bf v},\infty)$ is a $[7,3,4]_8$ NMDS almost self-dual code. 
    Verified by Magma \cite{Magma}, these results are true. 
\end{example}

\begin{example}\label{exam.asd2}
    Let $p=5$, $r=2$ and $m=4$. 
    Let $\omega$ be a primitive element of $\F_{5^4}$ satisfying $\w^4+4\w^2+4\w+2=0$ and 
    $\SSS=\{1,\omega^{26},\omega^{52},\omega^{78},\omega^{104},\omega^{130}, \omega^{182}, 2, \omega^{494}, \omega^{598}\}\subseteq \F_{5^2}\subseteq \F_{5^4}.$
    Note that $\sum_{a\in\SSS}a= \omega^{338}$ and $\SSS$ contains a $5$-zero-sum subset such as  $\{1, \omega^{52}, \omega^{78}, \omega^{494}, \omega^{598}\}$.  
    Take $v_i=(\lambda u_i)^\frac{1}{2}$ for each $1\leq i\leq 10$, where $\lambda=(-\sum_{a\in\SSS}a)^{-1}=\omega^{598}$, 
    and hence, we get that ${\bf v}=(\omega^{247},\omega^{260},\omega^{208},\omega^{247},\omega^{143},\omega^{39},\omega^{195},\omega^{26},\omega^{390},\omega^{65})\\ \subseteq (\F_{5^4}^*)^{10}.$
    By Theorems \ref{th.MDS condition} 2) and \ref{th.ASD2}, $\C_{k}(\SSS, {\bf v},\infty)$ is an $[11,5,6]_{5^4}$ NMDS almost self-dual code. 
    Verified by Magma \cite{Magma}, these results are true. 
\end{example}

\section{Concluding remarks}\label{sec.concluding remarks}
{
Motivated by {\bf Problems \ref{prob1}} and {\bf \ref{prob2}}, we {introduced} and {studied} a special class of linear codes $\C_k(\SSS,{\bf v},\infty)$ 
involving their parameters, non-GRS properties, weight distributions, and self-orthogonal properties in this paper. 
Specifically speaking, we {showed} that such codes must be non-GRS MDS codes or NMDS codes, 
and {gave} some related sufficient and necessary conditions. 
We also completely determine their weight distributions 
with the help of the solutions to subset sum problems. 
In particular, by comparing weight distributions, we {confirmed} that this class of linear codes is also not equivalent to 
the family of (+)-extended TGRS codes studied in \cite{ZL2024} and {derived} several new infinite families of $5$- and $6$-dimensional 
NMDS codes that are not equivalent to those constructed in \cite{HW2023}.   
Finally, based on the characterization of the self-orthogonal properties, 
we {showed} that there are no self-dual codes in this class of linear codes and 
{presented} two explicit constructions of almost self-dual codes. 
All these results gave answers to {\bf Problems \ref{prob1}} and {\bf \ref{prob2}}.

{Note also that the authors in \cite{HW2023,LZM2024} proved that minimum weight codewords of the linear codes constructed in \cite{HZ2023} 
can support $2$-designs in some cases. 
According to our Magma \cite{Magma} experiments, the minimum weight codewords of $\C_k(\SSS,{\bf v},\infty)$ 
can not support nontrivial $t$-designs in general, where $t\geq 2$. 
It would be interesting to further test if they can support nontrivial $t$-designs in some special cases.} 
On the other hand, as shown in Remark \ref{rem.1} and Table \ref{tab:newadd1}, it is verified by Magma \cite{Magma} that 
new NMDS codes are generally not derivable from $\C_k(\SSS,{\bf v},\infty)_\mu$ for $1\leq \mu \leq k-1$. 
Furthermore, in many cases, $\C_k(\SSS,{\bf v},\infty)_2$ is equivalent to $\C_k(\SSS,{\bf v},\infty)$.
Therefore, it would also be interesting to theoretically determine when they are equivalent 
or employ other methods to construct more families of linear codes that are either non-GRS  MDS codes or NMDS codes. 
}

\vfill


\begin{thebibliography}{1}
\bibliographystyle{IEEEtran}

%\bibitem{ACLMRS2022} Anderson S.E., Camps-Moreno E., L\'opez H.H., Matthews G.L., Ruano D., Soprunov, I.: Relative hulls and quantum codes. (2022) arXiv preprint arXiv:2212.14521. 
%\bibitem{AL2005} Abatangelo V., Larato B.: Near-MDS codes arising from algebraic curves. Discret. Math. {\bf 301}, 5-19 (2005).
\bibitem{B2015}  Ball S.: Finite geometry and combinatorial applications. Cambridge University Press, Cambridge (2015). 
\bibitem{BBPR2018} Beelen P., Bossert M., Puchinger S., Rosenkilde J.: Structural properties of twisted Reed-Solomon codes with applications to cryptography. in: Proc. IEEE Int. Symp. Inf. Theory (ISIT), pp. 946-950 (2018).
\bibitem{BPR2022}  Beelen P., Puchinger S., Rosenkilde J.: Twisted Reed-Solomon codes. IEEE Trans. Inf. Theory {\bf 68}(5), 3047-3061 (2022). 
\bibitem{Magma} Bosma W., Cannon J., Playoust C.: The Magma algebra system I: The user language. J. Symb. Comput. {\bf 24}(3-4), 235-265 (1997).


\bibitem{CHL2011} Cadambe V.R., Huang C., Li J.: Permutation code: Optimal exact-repair of a single failed node in MDS code based distributed storage systems. in: Proc. IEEE Int. Symp. Inf. Theory (ISIT), pp. 1225-1229 (2011).
{\bibitem{CHao2023} Chen H.: Many non-Reed-Solomon type MDS codes from arbitrary genus algebraic curves. IEEE Trans. Inf. Theory {\bf 70}(7), 4856-4864 (2024).}
{\bibitem{C2023} Cheng W.: On parity-check matrices of twisted generalized Reed-Solomon codes. IEEE Trans. Inf. Theory  {\bf 70}(5), 3213-3225 (2024).} 



\bibitem{defect} De Boer M.A.: Almost MDS codes. Des. Codes Cryptogr. {\bf 9}(2), 143-155 (1996). 
\bibitem{DT2020} Ding C., Tang C.: Infinite families of near MDS codes holding $t$-designs. IEEE Trans. Inf. Theory {\bf 66}(9), 5419-5428 (2020). 
{\bibitem{DTbook}Ding C., Tang C.: {\em Designs from linear codes}, second edition, World Scientific, Singapore, 2022.}
\bibitem{DL1994} Dodunekov S.,  Landjev I.:  On near-MDS codes. J. Geom. {\bf 54}(1-2), 30-43 (1994). 


\bibitem{FLL2020} Fang X., Liu M., Luo J.: New MDS Euclidean self-orthogonal codes. IEEE Trans. Inf. Theory {\bf 67}(1), 130-137 (2020). 
{\bibitem{FXF2021}Fang W., Xia S.-T., Fu F.-W.:  Construction of MDS Euclidean Self-Dual Codes via Two Subsets. IEEE Trans. Inf. Theory {\bf 67}(8), 5005-5015 (2021).}


%\bibitem{FL2022} Fu Y., Liu H.: Galois self-orthogonal constacyclic codes over finite fields. Des. Codes  Cryptogr. {\bf 90}(11), 2703-2733 (2022). 
%\bibitem{FL2023} Fu Y., Liu H.: Galois self-dual extended duadic constacyclic codes. Discrete Math. {\bf 346}(1), 113167 (2023).

\bibitem{GZ223} Gu H., Zhang J.: On twisted generalized Reed-Solomon codes with $\ell$ twists. IEEE Trans. Inf. Theory {\bf 70}(1), 145-153 (2023). 
%\bibitem{GJG2018} Guenda K., Jitman S., Gulliver T.A.: Constructions of good entanglement-assisted quantum error correcting codes. Des. Codes Cryptogr. {\bf 86}(1), 121-136 (2018).
\bibitem{GLLS2023} Guo G., Li R., Liu Y., Song H.: Duality of generalized twisted Reed-Solomon codes and Hermitian self-dual MDS or NMDS codes. Cryptogr.  Commun. {\bf 15}(2), 383-395 (2023).



\bibitem{HF2023} Han D., Fan C.: Roth-Lempel NMDS codes of non-elliptic-curve type. IEEE Trans. Inf. Theory {\bf 69}(9), 5670-5675 (2023).
\bibitem{HZ2023} Han D., Zhang H.: Explicit constructions of NMDS self-dual codes. Des. Codes Cryptogr. {\bf 92}, 3573–3585 (2024).   
%\bibitem{C2023}  Chen H.: New MDS entanglement-assisted quantum codes from MDS Hermitian self-orthogonal codes. Des. Codes Cryptogr. {\bf 91}, 2665-2676 (2023). 
\bibitem{HL2023} He B., Liao Q.: The properties and the error-correcting pair for lengthened GRS codes. Des. Codes Cryptogr. {{\bf 92}(1), 211-225 (2024).} 
\bibitem{H1929} Heineman E.R.: Generalized Vandermonde determinants. Trans. Amer. Math. Soc. {\bf 31}(3), 464-476 (1929). 
\bibitem{HLW2022} Heng Z., Li C., Wang X.: Constructions of MDS, near MDS and almost MDS codes from cyclic subgroups of $\F_{q^2}^*$. IEEE Trans. Inf. Theory {\bf 68}(12), 7817-7831 (2022).
\bibitem{HW2023} Heng Z., Wang X.: New infinite families of near MDS codes holding $t$-designs. Discret. Math. {\bf 346}(10), 113538 (2023). 
\bibitem{HYNYL2021}Huang D., Yue Q., Niu Y., Li X.: MDS or NMDS self-dual codes from twisted generalized Reed-Solomon codes. Des. Codes Cryptogr. {\bf 89}(9), 2195-2209 (2021).
\bibitem{HP2003} Huffman W.C., Pless V.:  Fundamentals of Error-Correcting Codes. Cambridge University Press, Cambridge (2003).

\bibitem{JK2019} Jin L., Kan H.: Self-dual near MDS codes from elliptic curves. IEEE Trans. Inf. Theory {\bf 65}(4), 2166-2170 (2019).
\bibitem{JX2014} Jin L., Xing C.: A construction of new quantum MDS codes. IEEE Trans. Inf. Theory {\bf 60}(5), 2921-2925 (2014).
\bibitem{JX2016} Jin L., Xing C.: New MDS self-dual codes from generalized Reed—Solomon codes. IEEE Trans. Inf. Theory {\bf 63}(3), 1434-1438 (2016).

\bibitem{LR2020} Lavauzelle J.,  Renner J.: Cryptanalysis of a system based on twisted Reed-Solomon codes. Des. Codes Cryptogr. {\bf 88}(7), 1285-1300 (2020).
\bibitem{LW2008} Li J., Wan D.: On the subset sum problem over finite fields. Finite Fields Appl. {\bf14}(4), 911-929 (2008). 
\bibitem{LZM2024} Li Y., Zhu S., Mart\'inez-Moro E.: On $\ell$-MDS codes and a conjecture on infinite families of 1-MDS codes. IEEE Trans. Inf. Theory {\bf 70}(10), 6899-6911 (2024).
\bibitem{LXW2008} Li Z., Xing L., Wang X.: Quantum generalized Reed-Solomon codes: Unified framework for quantum MDS codes. Phys. Rev. A {\bf 77}(1), 012308 (2008). 
\bibitem{ll2021} Liu H., Liu S.: Construction of MDS twisted Reed-Solomon codes and LCD MDS codes. Des. Codes Cryptogr. {\bf 89}, 2051-2065 (2021). 
\bibitem{LP2020} Liu H., Pan X.: Galois hulls of linear codes over finite fields.  Des. Codes Cryptogr. {\bf 88}(2), 241-255 (2020). 
\bibitem{LDMTT2021} Liu Q., Ding C., Mesnager S., Tang C., Tonchev V.D.: On infinite families of narrow-sense antiprimitive BCH codes admitting $3$-transitive automorphism groups and their consequences. IEEE Trans. Inf. Theory {\bf 68}(5), 3096-3107 (2022). 


%\bibitem{MC2021} Mi J., Cao X.: Constructing MDS Galois self-dual constacyclic codes over finite fields. Discret. Math. {\bf 344}(6), 112388 (2021). 
\bibitem{MMP2013} M\'arquez-Corbella I., Mart\'inez-Moro E., Pellikaan R.: The non-gap sequence of a subcode of a generalized Reed-Solomon code. Des. Codes Cryptogr. {\bf 66}, 317-333 (2013). 


{\bibitem{GRSequvialenceEGRS} Pellikaan R.: On the existence of error-correcting pairs. J. Stat. Plan. Inference {\bf 51}(2), 229-242 (1996).}


\bibitem{RL1989} Roth R.M., Lempel A.: A construction of non-Reed-Solomon type MDS codes. IEEE Trans. Inf. Theory {\bf 35}(3), 655-657 (1989). 

\bibitem{ST2013} Sakakibara K., Taketsugu J.: Application of random relaying of partitioned MDS codeword block to persistent relay CSMA over random error channels. in: Proc. $5$th Int. Congr. Ultra Modern Telecommun. Control Syst. Workshops (ICUMT), pp. 106-112 (2013). 
\bibitem{SV2018} Simos D.E., Varbanov Z.: MDS codes, NMDS codes and their secret-sharing schemes. in: $18$th International Conference on Applications of Computer Algebra (ACA' 12), Sofia, Bulgaria, (2012).
\bibitem{SYLH2022} Sui J., Yue Q., Li X., Huang D.: MDS, near-MDS or $2$-MDS self-dual codes via twisted generalized Reed-Solomon codes. IEEE Trans. Inf. Theory {\bf 68}(12), 7832-7841 (2022). 
\bibitem{SYS2023} Sui J., Yue Q., Sun F.: New constructions of self-dual codes via twisted generalized Reed-Solomon codes. Cryptogr. Commun. {\bf 15}(5), 959-978 (2023). 
\bibitem{SD2023} Sun Z., Ding C.: The extended codes of a family of reversible MDS cyclic codes. IEEE Trans. Inf. Theory  {\bf 70}(7), 4808-4822 (2024).
{\bibitem{SDC2024} Sun Z., Ding C., Chen T.: The extended codes of some linear codes. Finite Fields Appl. {\bf 96}, 102401 (2024).}


\bibitem{TR2018} Thomas A., Rajan B.: Binary informed source codes and index codes using certain near-MDS codes. IEEE Trans. Commun. {\bf 66}(5), 2181-2190 (2018). 


\bibitem{WLZ2023} Wan R., Li Y., Zhu S.: New MDS Self-Dual Codes Over Finite Field $\F_{r^2}$. IEEE Trans. Inf. Theory {\bf 69}(8), 5009-5016 (2023).
\bibitem{W2021} Wu Y.: Twisted Reed-Solomon codes with one-dimensional hull. IEEE Commun. Lett. {\bf 25}(2), 383-386 (2021). 
\bibitem{WHL2021} Wu Y., Hyun J.Y., Lee Y.: New LCD MDS codes of non-Reed-Solomon type. IEEE Trans. Inf. Theory {\bf 67}(8), 5069-5078 (2021).


%{\bibitem{Y2019} Yan H.: A note on the constructions of MDS self-dual codes. Cryptogr. Commun. {\bf 11}, 259-268 (2019).}
%{\bibitem{YZ2024} Yan Q., Zhou J.: Mutually disjoint Steiner systems from BCH codes. Des. Codes Cryptogr. {\bf 92}(4), 885-907 (2024).}


{\bibitem{ZF2020}Zhang A., Feng K.: A unified approach to construct MDS self-dual codes via Reed-Solomon codes. IEEE Trans. Inf. Theory {\bf 66}(6), 3650-3656 (2020).}
\bibitem{ZZT2022} Zhang J., Zhou Z., Tang C.: A class of twisted generalized Reed-Solomon codes. Des. Codes Cryptogr. {\bf 90}(7), 1649-1658 (2022). 
\bibitem{ZWXLQY2009} Zhou Y., Wang F., Xin Y., Luo S., Qing S., Yang Y.: A secret sharing scheme based on near-MDS codes. in: Proc. IC-NIDC, pp. 833-836 (2009). 
%{ \bibitem{Zhu2022} Zhu C.: The equivalence of GRS codes and EGRS codes. (2022) \text{https://doi.org/10.48550/arXiv.2204.11960}.}
\bibitem{ZL2024} Zhu C., Liao Q.: The (+)-extended twisted generalized Reed-Solomon code. Discret. Math. {\bf 347}(2), 113749 (2024). 








\end{thebibliography}
\end{document}